%% file: main.tex
\documentclass[journal]{IEEEtran}

\usepackage{commands_and_envs}

\begin{document}

\input{titlepage}
\input{abstract}
\input{content}
\input{appendix}
\input{acknowledgements}
\input{bibliography}

\end{document}

%% file: titlepage.tex
\title{Towards an Optimally Distributed Quantum Fourier Transform Circuit}

\author{
    \IEEEauthorblockN{
        Zachary~Vernec\thanks{This work was primarily conducted while Z.~Vernec was with the Department of Physics.},
        Michael~Silver,
        and Hans-Arno~Jacobsen
    } \\
    \IEEEauthorblockA{
        \textit{Edward S. Rogers Sr. Department of Electrical \& Computer Engineering} \\
        \textit{University of Toronto} \\
        Toronto, Canada \\
        \{zachary.vernec, m.silver\}@mail.utoronto.ca, jacobsen@eecg.toronto.edu
    }
}

\pagenumbering{gobble} 

\maketitle

%% file: abstract.tex
\begin{abstract}
A promising avenue for scaling quantum computing is to connect quantum processing units (QPUs) by generating entanglement between them. 
This requires \emph{circuit partitioning}: partially rewriting quantum circuits to run on a distributed quantum system using quantum teleportation protocols, while preserving the unitary operation implemented by the circuit.
The key metric to minimize when partitioning is the \emph{e-bit count}, defined as the number of maximally entangled qubit pairs that must be generated between QPUs.
We focus on partitioning the quantum Fourier transform (QFT) circuit, which is widely used as a subroutine in quantum algorithms such as quantum phase estimation and arithmetic circuits. 
Specifically, we present an optimal gate-packing-based partitioning scheme under bounded communication ancilla constraints, as well as a novel detached-gate-based partitioning scheme that achieves optimal e-bit consumption when ancilla availability is unrestricted.
We compare these schemes against prior analytical partitioning approaches for the QFT and evaluate them against partitions produced by general-purpose circuit partitioning algorithms.
We further validate our approach by implementing the bounded-ancilla circuits on quantum hardware.

\end{abstract}

\begin{IEEEkeywords}
Quantum Computing, QFT, Distributed Quantum Computing, Circuit Partitioning.
\end{IEEEkeywords}

%% file: content.tex
\section{Introduction} \label{sec:introduction}

\subsection{Motivation}

\IEEEPARstart{D}{istributed} quantum computing can be applied to solve the scaling issues that arise when building larger, more powerful quantum computing systems~\cite{caleffi_distributed_2024},~\cite{barral_review_2025}.
Indeed, while QPUs need many qubits to run many promising quantum algorithms such as Shor's, there are engineering difficulties in building a monolithic QPU at that scale~\cite{almudever_engineering_2017, fellous-asiani_optimizing_2023}, which has led quantum practitioners to consider connecting multiple quantum processing units (QPUs) into a distributed quantum system~\cite{cuomo_towards_2020,caleffi_distributed_2024,barral_review_2025}.
When multiple QPUs are connected using simple quantum channels that create entangled states across different QPUs, distributed quantum computing can be achieved by \emph{circuit partitioning}, where the circuit to run in a distributed manner is transformed to use mid-circuit measurement and classically controlled quantum gates when processing quantum information across QPU boundaries~\cite{baker_time-sliced_2020}.
As the entanglement creation operations that connect the QPUs are noisy and/or unreliable~\cite{main_distributed_2025,almanakly_deterministic_2025}, special care must be taken to ensure that the quantum circuit partitioning procedure is efficient in the amount of entanglement resource needed across QPUs.

\subsection{Problem Statement}

This paper describes a novel method for partitioning the quantum Fourier transform (QFT) circuit and collects, evaluates, and compares it with various other partitioning methods found in the literature.

Formally, an instance of our QFT partitioning problem is a set of $m$ QPUs with qubit capacities $n_1 \dots n_m$, pairwise connected by entanglement-generation channels.
A solution consists of (1) an assignment of the $n = n_1 + \dots + n_m$ input qubits of the QFT to the QPUs, respecting each QPU's capacity, and (2) a distributed circuit that implements the QFT under this assignment, in which the only operations crossing QPU boundaries are entanglement generation and classical communication (the ingredients of the quantum teleportation protocols described in Section~\ref{sec:background}).
Note that the QFT we are considering has the same number of qubits as the total system capacity, so all QPUs are fully filled\footnote{For QFTs needing fewer qubits than the system capacity, it is trivial to discard extraneous qubits and QPUs; the formulas from this paper will continue to hold.}.
The objective is to minimize the number of \emph{e-bits} (EPR pairs $\ket{\Phi^+}$) consumed by these teleportation protocols in the distributed circuit, as they are the dominant source of errors in distributed quantum computation.

Since the QFT circuit is often used as a subcircuit of larger quantum algorithms~\cite{harrow_quantum_2009, montanaro_quantum_2015, ruiz-perez_quantum_2017, wright_automatic_2021, bagherimehrab_nearly_2022, wright_noisy_2024, kahanamoku-meyer_fast_2024}, our findings can be used by any quantum computing practitioner who needs to partition a circuit that includes a QFT subcircuit, allowing them to choose the state-of-the-art circuit partition without needing their own computationally expensive analysis.

Indeed, while there are algorithms for compiling arbitrary quantum circuits to a distributed architecture~\cite{andres-martinez_automated_2019, andres-martinez_distributing_2024, burt_generalised_2024, burt_multilevel_2026, baker_time-sliced_2020, crampton_genetic_2024, g_sundaram_efficient_2021}, due to the generic nature of these compilers as well as the inherent complexity of compiling for distributed quantum computing~\cite{andres-martinez_distributing_2024}, there are few guarantees on the entanglement efficiency of the compilation results. The entanglement efficiency may be limited by the choice of compiler and by its compilation time.

Furthermore, we hope that by exploring various analytical schemes and computational algorithms for partitioning the QFT circuit, we can foster further insights into how best to develop strategies for analytically partitioning other circuits and how best to build new algorithms for quantum circuit partitioning.

\subsection{Approach} 

The approach taken is fourfold.
First, we highlight a gate-packing-based QFT circuit partition that has not been explicitly described before and prove its optimality with respect to a particular choice of gates and packings under the constraints of using a single ancilla qubit for entanglement generation.
Gate packing, as will be described in Section~\ref{sec:background}, refers to the idea of distributing more than one two-qubit gate using a single use of the gate teleportation protocol (and thus a single e-bit).  
Second, we explicitly describe a novel detached-gate-based QFT circuit partition and prove its optimality when allowing for an unbounded number of ancilla qubits for entanglement generation.
Furthermore, we compare our QFT circuit partitioning schemes with prior approaches from the literature, which have been developed in isolation and have not been compared to each other.
Some of these previous methods require generalization to be properly compared, and some need to be examined to find constant factors rather than asymptotic expressions.
Third, we evaluate the performance of circuit partitioning algorithms from the recent literature when given the QFT circuit as input, and find some evidence that our gate-packed QFT partition is better than the output of many partitioning algorithms in terms of required entanglement.
We also discuss how one algorithm obtains lower required entanglement, and the trade-offs it makes to obtain such a result.
Fourth, we proceed with implementing the partitioned QFT circuits, studying their behaviours on quantum computer simulators and on physical devices. 
In particular, we consider how the fidelity of a QFT circuit changes when implemented in a distributed manner.

There are a few limitations to the scope of our approach.
First, we only consider partitioning the standard, exact QFT circuit.
This means excluding alternative architectures such as the fast parallel QFT~\cite{cleve_fast_2000}, which substantially alter the standard circuit structure to reduce gate depth at impractically large scales~\cite{van_meter_fast_2005}.
It also means ignoring variations of the QFT such as the approximate QFT~\cite{coppersmith_approximate_1994} or the optimistic QFT~\cite{kahanamoku-meyer_log-depth_2025}. 
Even though these variations may be useful in practice, they either have a massive increase in complexity or change the properties of the computed operation (e.g., an approximate QFT is not equivalent to a traditional QFT), and have never been studied in the context of distributed quantum computing.
Second, we consider only the extreme cases: where each QPU is allowed a constant number of ancillas for entanglement generation, specifically either 1 or 2 per QPU, or where each QPU is allowed an unbounded number of such ancillas. 
Restricting the number of ancillas is sensible in the context of distributed quantum computing, where we are working under the assumption that scaling QPU sizes is difficult, but the specific restriction of one ancilla per QPU is potentially too restrictive.
While we also consider the unbounded case, we do not have a way to quantify the exact number of communication ancillas for our detached-gate-based distributed QFT, and we have not studied how to interpolate between the two extreme cases we present.
Third, we do not consider the possibility of state teleportation as a circuit distribution primitive. 
However, we conjecture that this would not improve e-bit counts for our QFT circuit definition which excludes the SWAP gates.

\subsection{Contributions}

The contributions of this research are fivefold.

First, we give an explicit description of a scheme for partitioning the QFT through the new circuit partitioning techniques known as \emph{gate packing}. 
We know this partition is optimal for two QPUs as it achieves the lower bounds on e-bit count given by quantum communication complexity.
For the case of more than two QPUs, we include a proof that this scheme is optimal when restricting to gate-packing-based circuit partitioning of the standard QFT circuit, without detached gates.

Second, we give a scheme for partitioning the QFT when also using \emph{detached gates}, which we also show is optimal when restricting to gate-teleportation-based circuit partitioning of the standard QFT circuit, this time when allowing unrestricted communication ancillas used to store entanglement between QPUs.

Third, we survey explicit partitioning schemes for the QFT, analyze them in full generality, and compare them to our proposed approach.
This enables those interested in partitioning a circuit involving the QFT to choose their preferred method by understanding the exact requirements in using different partitioning schemes for any number and size of QPUs.

Fourth, we compare our two distributed QFT schemes with the distributed QFT circuits output by circuit partitioners from the literature.
These algorithms are designed to take arbitrary circuits as input, in contrast to our gate-packing scheme designed specifically for the QFT circuit.
This provides evidence that gate packing gives an optimal partition of the QFT circuit when a single ancilla per QPU is used, and that state teleportation does not lower the minimum e-bits required for partitioning the QFT circuit under consideration here.
It also shows that all circuit partitioners remain far from the detached-gate optimum when targeting more than two QPUs.

These results show the advantage of defining explicit partitioning schemes for the QFT as opposed to relying on circuit partitioners.
Any practitioner needing to partition the QFT circuit, if they have access to only one ancilla qubit per QPU for entanglement generation, would benefit from implementing the gate-packing-based partitioning scheme directly rather than running a circuit partitioner: they would save compilation time while obtaining the best result.
Furthermore, these results would hold when the practitioner has to partition a circuit which includes the QFT as a subcircuit of a larger quantum algorithm.

Fifth, we perform a preliminary investigation on the fidelity of results when implementing selected partitions of the QFT circuit in simulations and on hardware.
This shows that, without error mitigation, a non-partitioned QFT circuit is no different from various partitioned QFT circuits.
However, it also shows that more efficient circuit partitions are faster to run, and this improvement grows with circuit size.
This gives a starting point for future research into comparing implementations of circuit distributions of the QFT.

\subsection{Organization}

Section~\ref{sec:background} presents some background notation and terminology relating to circuit partitioning.
Section~\ref{sec:related-work} examines related work in partitioning the QFT, from previously defined QFT circuit partitioning schemes, to theoretical work establishing bounds on distributing the QFT unitary across two QPUs, to algorithms for partitioning arbitrary circuits, and to a methodology for evaluating any particular QFT implementation (whether distributed or not). 
Section~\ref{sec:content-proposed-dqft-packed} examines our first partitioning scheme, a gate-packing-based circuit partitioning of the QFT circuit, finds an analytical expression for the e-bit requirements, and proves minimality when comparing to similar circuits. 
Section~\ref{sec:content-proposed-dqft-detached} examines our second partitioning scheme, a gate-packing and detached-gate-based circuit partitioning of the QFT circuit, and proves its minimality in this more generalized setting. 
Section~\ref{sec:analysis-previous} elaborates on the e-bit requirements of previous partitioned QFT circuits to compare them with the gate-packing-based QFT circuit partitioning scheme established and analyzed in the preceding section.
Section~\ref{sec:content-empirical} compares our two QFT partitioning schemes to partitions of the QFT found by circuit partitioners.
Section~\ref{sec:content-implementation} discusses the challenges in implementing our gate-packed QFT in software (through simulation) and on hardware (through IBM QPUs), the solutions developed for surmounting these challenges, and the methodology used for evaluating the fidelity of our implementation.

\section{Background} \label{sec:background}

\begin{figure*}[!t]
    \centering
    \includegraphics[width=0.80\linewidth]{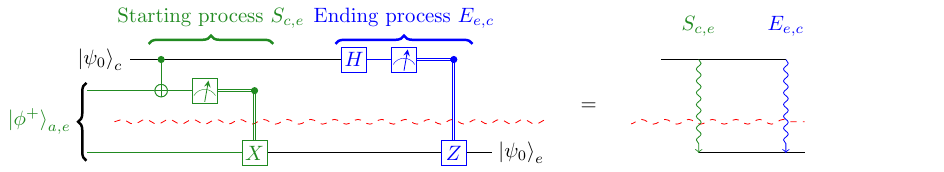}
    \caption[Teledata Circuit Diagram Under Wu~et~al.'s Framework]{A state teleportation protocol (teledata) using starting and ending processes. The LHS uses a common circuit notation such as would be used in~\cite{nielsen_quantum_2010}, while the RHS has the same circuit in the notation of~\textcite{wu_entanglement-efficient_2023} with snake arrows representing the flow of information.}
    \label{fig:Wu-teledata}
\end{figure*}

\begin{figure*}[!t]
    \centering
    \includegraphics[width=0.80\linewidth]{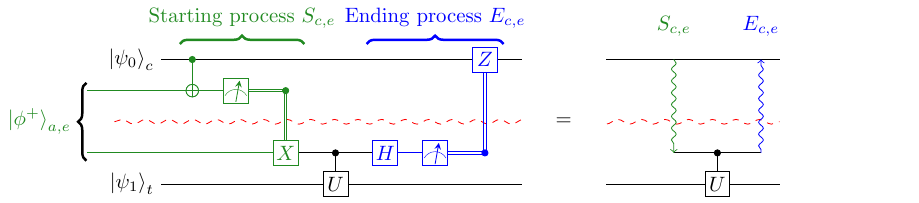}
    \caption[Telegate Example Circuit Diagram Under Wu~et~al.'s Framework]{A gate teleportation protocol (telegate) using starting and ending processes. The same diagrammatic notation is used on the left-hand side and right-hand side as Fig.~\ref{fig:Wu-teledata}.}
    \label{fig:Wu-telegate}
\end{figure*}

There are three essential components to a circuit partitioning framework: (1) entanglement generation, which creates the entanglement resource, (2) information distribution, which consumes the entanglement resource, and (3) information centralization, which ensures the information distribution only acts temporarily.
In both (2) and (3), \emph{mid-circuit measurement} and \emph{classically controlled quantum gates} are necessary.
This is when a qubit register is measured before the end of the circuit, and a following gate is applied conditionally with respect to the measurement outcome.
Circuits using mid-circuit measurement and classical control are also called \emph{dynamic circuits}.

However, as part of a separation of concerns, the exact process for entanglement generation (also known as entanglement distribution) is discussed as a separate abstraction layer, more closely related to quantum networking.
It is for this reason that entanglement generation will not be considered here; the information distribution component will be assumed to always have access to the necessary entanglement.

In this paper, we use the circuit partitioning framework of~\citeauthor{wu_entanglement-efficient_2023}.
The entanglement resource considered in this framework is that of \emph{e\nobreakdash-bits}, also known as EPR pairs, since they are a maximally entangled state that is straightforward to describe.
Furthermore, e-bits are a good abstraction considering that entanglement generation is a physical process: for example, spontaneous parametric down-conversion can be used to produce Bell pairs~\cite{shin_bell_2019}, which are equivalent to e-bits up to simple local operations.

Now, the entanglement generation and the information distribution steps are combined into one primitive called the \emph{entanglement assisted starting process} (or simply \emph{starting process}) $\Starting_{c_0,a_1}$ where the subscript indicates the process acts on qubit register $c_0$ and distributes the quantum information onto a new ancillary register $a_1$.
This primitive takes as input a qubit $\ket{\psi} = \alpha\ket{0}+\beta\ket{1}$ in some register $c_0$ of QPU $0$.
First, it generates an e-bit $\ket{\Phi^+} = \frac{1}{\sqrt{2}} \left( \ket{00} + \ket{11} \right)$ shared between ancillary registers $a_0$ of QPU~0 and $a_1$ of QPU~1.
These ancillary registers are called communication qubits.
Then, the process applies a $\CX$ gate with control $c_0$ and target $a_0$, followed by measuring $a_0$ giving a result $r \in \set{0,1}$, and finally using the measurement result $r$ to apply $\X^r$ to the other ancillary register $a_1$ (i.e., apply $\id$ if $r=0$ or $\X$ if $r=1$).
This has the effect
\begin{equation}
    \Starting_{c_0,a_1} : \ket{\psi}_{c_0} \mapsto \left( \alpha\ket{00}+\beta\ket{11} \right)_{c_0,a_1} \;,
\end{equation}
where the subscript indicates the qubit registers which hold each quantum state.
Note that the information contained in the parameters $\alpha,\beta$ defining $\ket{\psi}$ is now distributed across both QPUs through this resulting state.

Now, the \emph{ending process} primitive $\Ending_{a,k}$ is the information centralization step, acting as the inverse of the starting process to restore the quantum information to a non-distributed state. 
When acting on a state $\alpha\ket{00}+\beta\ket{11}$ that is identical to the output of a starting process, the ending process has the effect
\begin{equation}
    \Ending_{k,a} : \left( \alpha\ket{00}+\beta\ket{11} \right)_{a,k} \mapsto \ket{\psi}_{k} \;,
\end{equation}
where $\ket{\psi}$ is still defined as above with parameters $\alpha,\beta$.
Specifically, this process is defined as applying $\H$ on register $a$, measuring its result $r \in \set{0,1}$, and applying $Z^r$ to the register $k$.

The circuit definition of a starting process (in green) and an ending process (in blue) can be seen on the left-hand sides of Fig.~\ref{fig:Wu-teledata} or Fig.~\ref{fig:Wu-telegate}.

\begin{figure*}[!t]
    \centering
    \includegraphics[width=\linewidth]{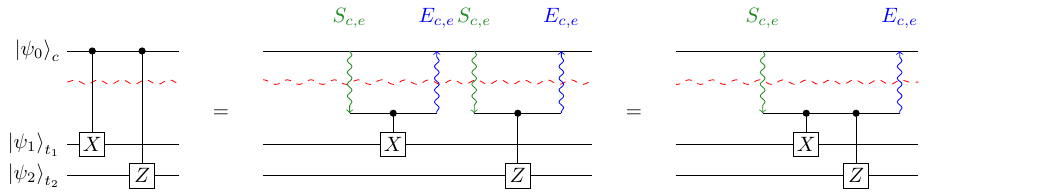}
    \caption[Example of Gate Packing]{An example of gate packing the $\CX$ gate and the $\CZ$ gate, both with control qubit $c$ and targeting qubits $t_1, t_2$ on the same QPU. On the left is the circuit that needs to be partitioned. In the middle is the circuit partition without gate packing, which requires two starting processes and therefore two e-bits. On the right is an equivalent partition where the $\CX$ and $\CZ$ gates are packed so that a single starting process (and single e-bit) is required.}
    \label{fig:Wu-packing}
\end{figure*}

\subsection{Quantum Teleportation Protocols (Teledata and Telegate)} \label{sec:Wu-teledata-telegate}

There are two protocols useful in circuit partitioning that can be implemented by combining the starting and ending process primitives.

The first is the \emph{teledata} protocol, also known as \emph{quantum state teleportation}~\cite{bennett_teleporting_1993}.
State teleportation from a state in register $c_0$ to register $a_1$ is defined as
\begin{equation} \label{eq:teledata}
    \TP_{c_0 \rightarrow a_1} := \Ending_{a_1,c_0} \of \Starting_{c_0,a_1} \;,
\end{equation}
where $\of$ denotes the composition of operators.
It acts as $\ket{\psi}_{c_0} \mapsto \ket{\psi}_{a_1}$.

The second is the \emph{telegate} protocol, also known as \emph{quantum gate teleportation}~\cite{eisert_optimal_2000}.
It is essentially an in-place teledata protocol ($\TP_{c_0 \rightarrow c_0} = \Ending_{a_1,c_0} \of \Starting_{c_0,a_1}$) but with the addition of an appropriate gate between the starting and ending processes.
Indeed, suppose we want to apply some gate $\U$ to act non-locally between register $c_0$ on some QPU and some set of target registers $T$ on some other QPU.
Suppose further that the effective action of $\U$ on register $c_0$ is the identity: that is, that
\begin{align*}
    \U_{c_0} \ket{0}_{c_0} &= \ket{0}_{c_0} \tensor U'_T \\
    \U_{c_0} \ket{1}_{c_0} &= \ket{1}_{c_0} \tensor U''_T
\end{align*}
for some arbitrary $\U',\U''$ acting on $T$.
A common example are controlled unitaries such as $\CX$, $\CZ$, $\gate{CXX}$ with the control being $c_0$, or $\gate{CCX}$ (Toffoli) with one of the controls being $c_0$.
Now, the telegate is defined as 
\begin{equation} \label{eq:telegate}
    U_{c_0,T} := \Ending_{c_0,a_1} \of \U_{a_1, T} \of  \Starting_{c_0,a_1} \;,
\end{equation}
where again $a_1$ is the entanglement-assisting ancilla register.
Note that since $\U$ doesn't directly affect the state on registers $c_0,a_1$, the ending process acts in the same way as before to recreate the state $\ket{\psi}_{c_0}$, while not affecting the entanglement created from applying $\U$.
Therefore, the telegate procedure described in Equation~\eqref{eq:telegate} is equivalent to applying the gate $\U$ non-locally between $c_0$ and $T$.

Finally, note that teledata and telegate can be combined by changing $\Ending_{c_0,a_1}$ to $\Ending_{a_1,c_0}$ in Equation~\eqref{eq:telegate}, denoted as $U_{c_0 \rightarrow a_1, T}$.

The left-hand sides of Fig.~\ref{fig:Wu-teledata} and Fig.~\ref{fig:Wu-telegate} show a circuit representation of the teledata and telegate protocols with the starting and ending processes and the e-bit clearly identified.
Furthermore, the right-hand side of these figures presents a simplified notation introduced in~\cite{wu_entanglement-efficient_2023} which hides the specific gates used in the implementation of the starting and ending processes.
This greatly helps in reducing visual clutter when describing how a larger circuit can be partitioned.

\subsection{Gate Packing} \label{sec:Wu-gate-packing}

An advantage of decomposing quantum teleportation protocols into starting and ending processes is that more complicated non-local operations can be built from these primitives.

The most relevant technique to this paper for building such a non-local operation is that of \emph{gate packing}.
This technique was hinted at by \citeauthor{yimsiriwattana_generalized_2004}~\cite{yimsiriwattana_generalized_2004} but not fully developed and identified as its own technique until the starting and ending processes framework was introduced by \citeauthor{wu_entanglement-efficient_2023}~\cite{wu_entanglement-efficient_2023}.

At the most basic level, gate packing involves identifying sequences of two-qubit gates that must be partitioned, possibly with single-qubit gates interspersed, and partitioning the entire sequence using a single starting process and a single ending process bracketing the entire sequence.
This is done while ensuring that the new sequence of gates (that now includes the starting and ending processes) is equivalent at an abstract level to the original sequence, which might mean that the sequence of gates in between the starting and ending processes might be modified.
This can allow for some savings in e-bits needed for circuit distribution.

The intuition for this is that any given sequence of gates is simply some specific decomposition of a more complicated unitary.
However, this specific decomposition may not efficiently use e-bits as a resource if each component is partitioned separately.
Instead, packing the sequence of gates amounts to finding an alternate decomposition of the same complicated unitary, and since there are a large number of ways to decompose a complicated unitary, doing so may lead to finding a decomposition which is more efficient in terms of e-bit requirements. 

It is important to note that there are various conditions that a sequence of gates must satisfy to be packed together.
For details, see~\cite{wu_entanglement-efficient_2023}, in which some conditions for gate packing are given with proofs.
The most relevant case of gate packing for the circuit distribution of the QFT introduced in this paper is also the simplest: a sequence of controlled unitaries $CU^{(1)}\dots, CU^{(l)}$ whose control qubits are the same and whose target qubits are on the same QPU can be packed together as in Fig.~\ref{fig:Wu-packing}.
Note that this simple version of gate packing was first described in~\cite{yimsiriwattana_generalized_2004}, but was not identified as a distinct concept and was not generalized to more complicated gate sequences.

Another example of a more complicated technique that can be used for circuit distribution is that of the  \emph{detached gate}, as shown in Fig.~\ref{fig:detached-gate}.
Suppose we want to non-locally implement a gate $U$ between a control qubit~$a$ on QPU~$A$ and a target qubit~$b$ on QPU~$B$.
In the non-detached case, the telegate follows Equation~\eqref{eq:telegate}, with a single starting process from qubit~$a$ to an ancilla on QPU~$B$.
The gate~$U$ is implemented locally on QPU~$B$ between the ancilla and qubit $b$ before the ending process restores qubit~$a$.
In the detached case, however, there is a third QPU~$C$ involved; both qubits~$a$ and~$b$ undergo a starting process from their respective QPUs to QPU~$C$.
The gate~$U$ is now implemented locally on QPU~$C$ before both qubits~$a$ and~$b$ are restored by their respective ending processes.
Note that this requires two e-bits instead of one, which is worse in this case but may be advantageous in some circumstances due to gate packing.

\begin{figure}[!t]
    \centering
    \includegraphics[width=0.6\linewidth]{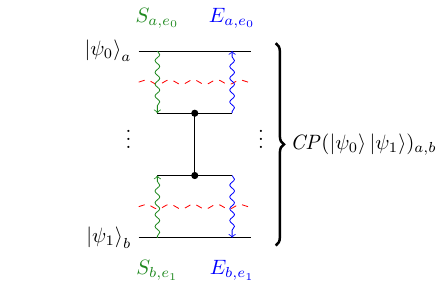}
    \caption[Example of a Detached Gate]{An example of a detached gate. A non-local $\CZ$ gate is applied between qubits of register $a$ (on QPU $A$) and register $b$ (on QPU $B$). However, due to both having a starting process to distribute towards a middle QPU $C$, the local $\CP$ gate is applied between qubits of QPU $C$.}
    \label{fig:detached-gate}
\end{figure}

\begin{figure*}
    \centering
    \includegraphics[width=0.85\linewidth]{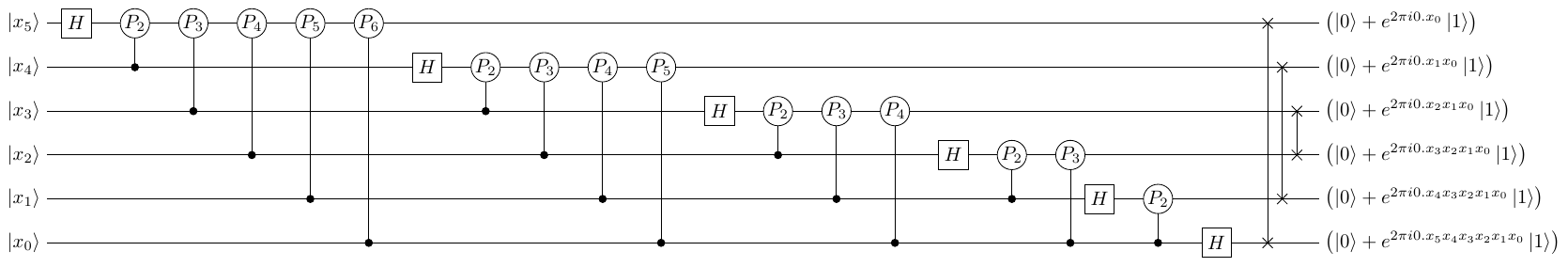}
    \caption[Traditional QFT Circuit Diagram]{The traditional circuit for the quantum Fourier transform on 6 qubits, where the subscript $k$ of a controlled phase gate defines a phase factor $e^{2\pi i/2^{k}}$ as in~\cite{nielsen_quantum_2010}.}
    \label{fig:traditional-qft}
\end{figure*}

\begin{figure*}[!t]
    \centering
    \includegraphics[width=0.85\textwidth]{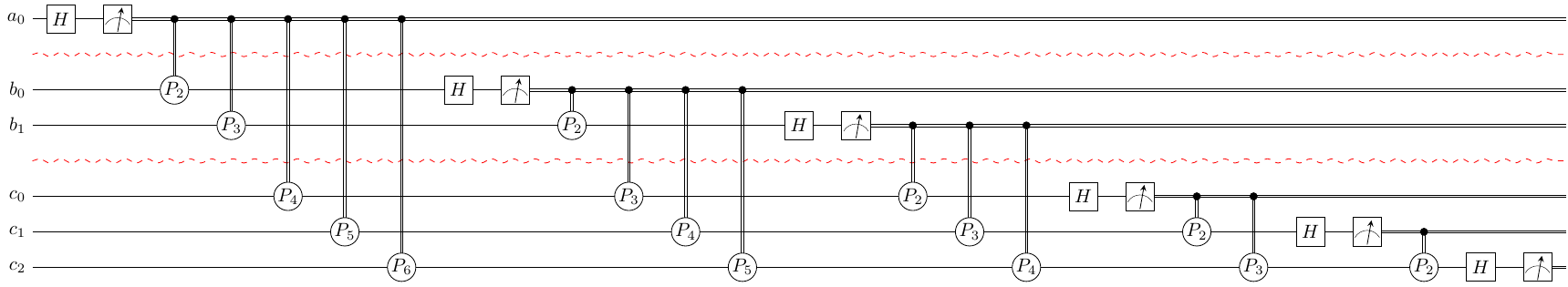}
    \caption[Semi-Classical QFT Circuit Diagram]{Semi-classical QFT circuit on 6 qubits. Note that the double lines represent classical information propagation and classical control.}
    \label{fig:semi-classical-qft}
\end{figure*}

\section{Related Work} \label{sec:related-work}

There are four main areas of study that relate to the present work: (1) analytical schemes to rewrite and/or distribute the QFT circuit, (2) theoretical bounds on efficiently distributing the QFT, (3) partitioning algorithms that have been benchmarked on the QFT circuit, and (4) a method for measuring the fidelity of a noisy QFT implementation.
We discuss these in order.

\subsection{Circuit Partitions of the QFT Circuit} \label{sec:other-distributed-qft}

The QFT circuit is traditionally defined as in Fig.~\ref{fig:traditional-qft}, with sequences of Hadamard gates $\H$, controlled phase gates $\gate{CP}$, and ending $\SWAP$ gates.
There are a few common variants of this circuit.

The first notable variant is the \emph{core} QFT, which is the QFT without the ending $\SWAP$ gates~\cite{chen_quantum_2023}.
Since the $\SWAP$ gates serve only to reorder the output into the correct endianness, this is not required to be applied using quantum operations.
The classical controller part of the quantum computer can keep track of the endianness of the circuit throughout its execution and, when a following gate needs to be applied on some qubits, instead apply the gate on where the qubits should be given the tracked endianness at that point.
Similarly, after measuring the circuit output, classically permute the resulting bits to match the permutation the SWAP gates would have effected.
This is a \emph{virtual} gate, where the classical preprocessing, controller, and postprocessing take over the role of a quantum gate.
Since virtual gates can be implemented classically without introducing errors, for the rest of this paper, the $\SWAP$ gates will not be shown in the circuit diagrams and will be assumed to be implemented virtually.

Another common variant of the QFT is the \emph{QFT+Measurement} operation~\cite{griffiths_semiclassical_1996}.
This is the (core) QFT circuit, which is immediately followed by a measurement of all qubits.
Note that many applications of the QFT (or inverse QFT) circuit are of this form, such as Quantum Phase Estimation (and consequently Shor's algorithm).
However, as the QFT+Measurement is a measurement operation instead of a unitary, it cannot be used coherently within a larger quantum circuit---a requirement for implementing QFT-based operations such as quantum-walk shifts and arithmetic circuits.

Following the QFT+Measurement research direction, \textcite{griffiths_semiclassical_1996} present the \emph{semi-classical} version of the QFT+Measurement circuit.
The circuit is as shown in Fig.~\ref{fig:semi-classical-qft}.
By the deferred measurement principle, it is equivalent to the core QFT circuit immediately followed by measuring all qubits.
Note that this gives a distributed circuit, with the advantage that all non-local operations are replaced by classically controlled single-qubit gates.
This means that the circuit can be distributed without the need for a quantum channel to pre-share e-bits.
However, it is limited in applicability due to only implementing the QFT+Measurement operation instead of the unitary QFT.

Alternatively, \textcite{yimsiriwattana_generalized_2004} present a circuit distribution for the unitary QFT using telegates, after having introduced telegates in their early circuit distribution framework.
The procedure is simple: for each two-qubit gate which must act non-locally, implement it using a telegate.
This can be seen in Fig.~\ref{fig:yimsiriwattana-qft}.
\begin{figure*}[!t]
    \centering
    \includegraphics[width=0.7\textwidth]{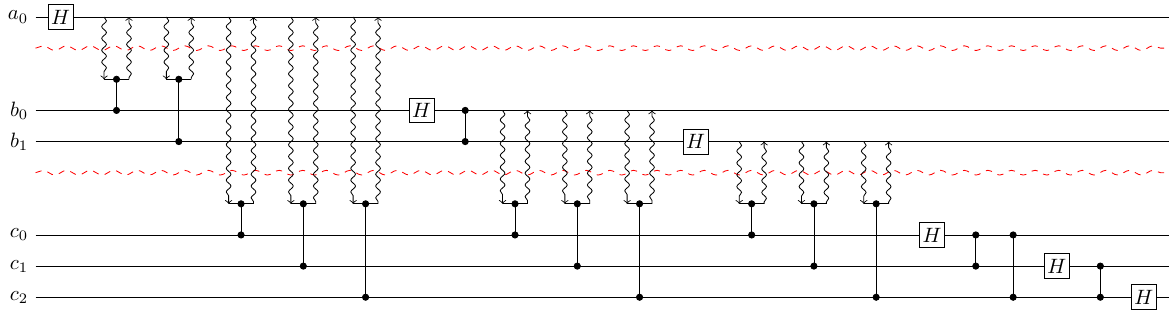}
    \caption[Yimsiriwattana and Lomonaco~Jr's Distributed QFT Circuit Diagram]{\citeauthor{yimsiriwattana_generalized_2004}'s distributed QFT circuit on 6 qubits across 3 uneven-sized QPUs, drawn using the starting and ending process notation. Note that for simplicity, the $\CP$ gate angles are not displayed.}
    \label{fig:yimsiriwattana-qft}
\end{figure*}
Note that this does not take advantage of gate packing.
Due to the simplicity of this implementation, we call it the \emph{naive} implementation.%
\footnote{Remark: This term is not used pejoratively and is not a criticism of the original work. Implementing a distributed QFT was not the main focus of~\cite{yimsiriwattana_generalized_2004}, and their circuit partitioning is one of the earliest examples in the literature of partitioning a full circuit.}

\begin{figure*}[!t]
    \centering
    \includegraphics[width=0.7\textwidth]{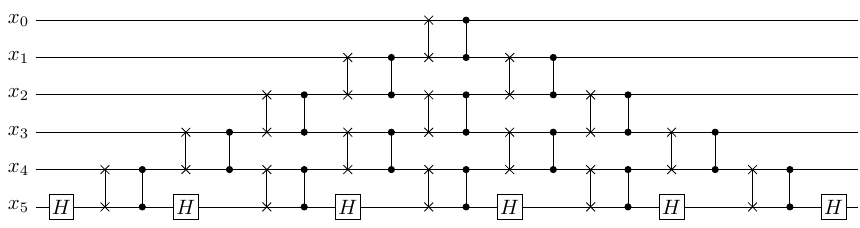}
    \caption[Linear Topology QFT Circuit Diagram]{The linear topology QFT on 6 qubits, as diagrammed in~\cite{jin_optimizing_2024}. Note the angles of the $\CP$ gates are not shown.}
    \label{fig:linear-qft}
\end{figure*}
\begin{figure*}[!t]
    \centering
    \includegraphics[width=0.7\textwidth]{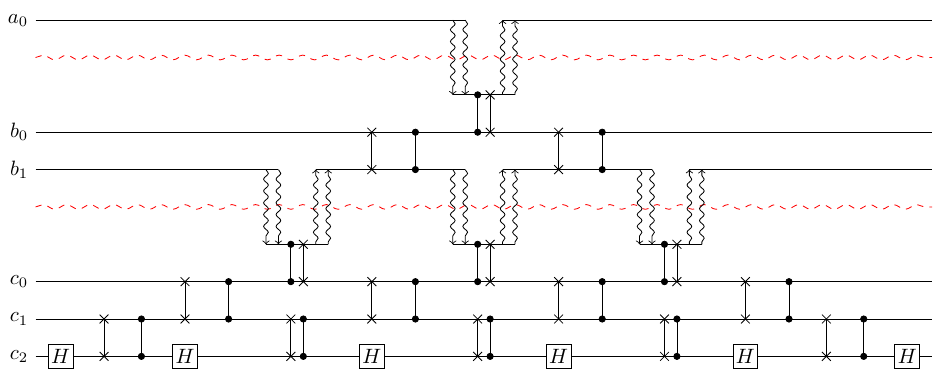}
    \caption[Partitioned Linear Topology QFT Circuit Diagram]{The linear topology QFT on 6 qubits, using teledata-based circuit distribution. Recall that the double snake arrows represent state teleportation.}
    \label{fig:linear-qft-distributed}
\end{figure*}

\textcite{rodney_van_meter_communications_2004} examines mapping the QFT across a distributed quantum system composed of interconnected linear topology QPUs, where the QPUs are evenly-sized linear rails.
His method consists of mapping to the full topology of the distributed quantum system (both including intra-device and inter-device topologies) using $\SWAP$ gates to ensure the circuit matches the topological constraints.
The paper also considers analytical expressions for the number of inter-device $\SWAP$ gates for distributing an $n$ qubit QFT evenly across $m$ QPUs, but considers the inter-device $\SWAP$ gates as a basic operation, falling short of counting the number of e-bits necessary.
Note that they consider both the \emph{return case} and \emph{no-return case}, which are equivalent to the QFT and core QFT, respectively.
Other topologies are also explored, such as a square-lattice-based system topology (each QPU being an edge) and multi-rail topologies (where each QPU is a rectangular array), though the distribution costs for these are not investigated.

\subsection{Quantum Communication Complexity of the QFT}

In~\cite{nielsen_quantum_2000}, the \emph{coherent communication complexity}~$Q_0$ of a quantum operation $U$ is defined as the minimum number of qubits that must be exchanged between two partitions to perform $U$.
Since any qubit communication can also be simulated using the same number of e-bits by implementing state teleportation instead, the coherent communication complexity is a lower bound for the number of e-bits required to distribute a circuit implementing $U$.
As summarized in~\cite{tyson_operator-schmidt_2003}, \citeauthor{nielsen_quantum_2000} establishes the bounds
\begin{equation} \label{eq:nielsen-bound}
     \ceil{\log_2(\Schmidt U)} \leq Q_0(U) \leq 2\min(n_1,n_2) \;,
\end{equation}
where $\Schmidt$ is the number of non-zero Schmidt coefficients needed in the Schmidt decomposition (for a given partition into $n_1,n_2$ qubits).
The lower bound relies on the Schmidt decomposition of the unitary $U$, while the upper bound is simply using teledata to send all qubits from the smaller QPU to the larger QPU, operating, and sending the qubits back.

Specific application to the QFT is summarized and generalized in~\cite{tyson_operator-schmidt_2003}, which provides a Schmidt decomposition for partitioning the $\QFT_{n_1+n_2}$ into the first $n_1$ qubits and last $n_2$ qubits and finds that $\Schmidt(\QFT_{n_1+n_2}) = \min\left({2^{2n_1}, 2^{2n_2}}\right)$.%
\footnote{In fact, a more general result is obtained by considering arbitrary qudits instead of only qubits, but we only consider qubits in our work.}
Therefore, by Equation~\eqref{eq:nielsen-bound},
\begin{equation} \label{eq:tyson-bound}
    Q_0(\QFT_{n_1+n_2}) = 2\min(n_1,n_2) \;.
\end{equation}
To emphasize, this implies a lower bound of $2\min(n_1,n_2)$ e-bits to circuit partition $\QFT_{n_1+n_2}$ across a bipartition $(A,B)$, where $A$ holds the first $n_1$ qubits and $B$ holds the last $n_2$ qubits.
It is also important to emphasize that the partition considered is highly ordered (sequential qubits), and an arbitrary partition $A,B$ of the same sizes $n_1,n_2$ could lead to a different Schmidt number and so a different lower bound. 
It is further important to note that this bound applies to the QFT operation as traditionally defined, which, when implemented as a circuit, needs many (non-virtual) SWAP gates at the end of the circuit.

The core QFT operation is substantially different.
\textcite{chen_quantum_2023} show that the entangling power of the QFT is mainly due to the ending SWAP operations: in contrast to the previous Schmidt decompositions used to show Equations~\eqref{eq:nielsen-bound} and~\eqref{eq:tyson-bound}, which have uniformly large Schmidt coefficients, the Schmidt decomposition of the core QFT in~\cite{chen_quantum_2023} gives Schmidt number $\Schmidt(\QFT_{n_1+n_2}) = \min\left(2^{n_1}, 2^{n_2}\right)$.

While \textcite{chen_quantum_2023} established their Schmidt number in the context of matrix product simulations of the QFT, they did not consider its relation to distributed quantum computing or communication complexity.
We now note that this Schmidt number can be used with the coherent communication complexity bound of Equation~\eqref{eq:nielsen-bound}, resulting in
\begin{equation} \label{eq:chen-bound}
     \min(n_1,n_2) \leq Q_0(\QFT_{n_1+n_2}) \leq 2\min(n_1,n_2) \;,
\end{equation}
where $\QFT$ in the above equation stands for the core QFT.

While the above bounds are important theoretical results, it is important to note that they are limited to the bipartite (two QPU) case.
In fact, since the Schmidt decomposition is not well-defined for multipartite systems, this method does not generalize to e-bit minimums for arbitrary partitions.
It is also important to recall that this is limited to a bipartition into sequential qubit blocks (first $n_1$ qubits vs. last $n_2$ qubits), and only the more general Equation~\eqref{eq:nielsen-bound} applies to an arbitrary bipartition.

\subsection{Circuit Partitioners} \label{sec:general-circuit-partitioners}
In this paper, we focus on the QFT circuit; however, others have considered the general problem of algorithmically distributing arbitrary quantum circuits.
These algorithms, which we call \emph{circuit partitioners} but may also be considered a type of \emph{compiler} for quantum circuit distribution, analyze the structure of the quantum circuit and return an equivalent quantum circuit by using telegates and teledata to implement non-local operations.
\textcite[Table~2]{caleffi_distributed_2024} cite a few recent methods, though only a few of them use the QFT circuit for benchmarking.
In particular,~\textcite{burt_generalised_2024,burt_multilevel_2026} use the QFT circuit as a benchmark for comparing their partitioners with previous literature.
Since we will use these partitioner results as benchmarks for our analytical circuit partitioning scheme, we will now briefly describe the circuit partitioners under consideration.

First, there is \citeauthor{baker_time-sliced_2020}'s \emph{Fine-Grained Partitioner using relaxed Overall Extreme Exchange}: FGP-rOEE.
This circuit partitioner computes a qubit-to-QPU assignment from a graph partition for each circuit time slice using a variant of the Overall Extreme Exchange (OEE) algorithm~\cite{park_algorithms_1995} (itself based on the Kernighan-Lin algorithm~\cite{kernighan_efficient_1970}) with an early stopping mechanism; this relaxation gives the algorithm its name, \emph{relaxed} OEE.
Between time-slices, qubits may move from one QPU to another, which is implemented using teledata.
No other circuit distribution techniques are considered by this algorithm.

Second, there are \citeauthor{andres-martinez_distributing_2024}'s circuit partitioning workflows~\cite{andres-martinez_distributing_2024} made available in the \emph{Pytket-DQC} package~\cite{andres-martinez_cqclpytket-dqc_2025}.
The main circuit representation used is a hypergraph~\cite{andres-martinez_automated_2019}, where hyperedges consist of a distributable packet, i.e., one qubit-vertex and any number of two-qubit gates that can be packed together with that qubit as control.
This hypergraph can be further transformed to take into account gate embedding (an advanced technique from~\cite{wu_entanglement-efficient_2024} that makes it possible to pack even more gates together), and to take into account detached gates (making it possible to implement a telegate on an unrelated QPU).
The three workflows highlighted in~\cite{burt_generalised_2024,burt_multilevel_2026} are \emph{AnnealingEmbedSteinerDetach} (Pytket-AESD), \emph{PartitionEmbed} (Pytket-PE), and \emph{PartitionHeteroEmbed} (Pytket-PHE).

The Pytket-AESD workflow uses simulated annealing~\cite{kirkpatrick_optimization_1983} for hypergraph partitioning, applies vertex covering to identify possible gate embeddings, and merges remaining distributable packets with gates that act on a common qubit.   
Finally, some gate vertices are reallocated in the hypergraph to take advantage of the detached gate trick.

In contrast, both the Pytket-PE and Pytket-PHE workflows start by using the standard hypergraph partitioner \mbox{KaHyPar}~\cite{schlag_high-quality_2023} on the initial hypergraph representing the circuit, followed by a refinement pass to take into account gate embedding when possible (according to the current hypergraph partition).
The Pytket-PHE workflow has an additional step of reallocating hypergraph vertices belonging to cut hyperedges to minimize e-bit count, greedily and sequentially, until no more improvement is possible.

Note that in all three workflows, teledata is not considered as a distribution method, but the full breadth of telegate transformations is considered, which includes gate packing, gate embedding, and detached gates.

Third, there are \citeauthor{burt_generalised_2024}'s \emph{Generalized Circuit Partitioners}~\cite{burt_generalised_2024}, \emph{Simple} and \emph{Extended} variants: GCP-S and GCP-E.
The circuit representation used is that of the interaction graph extended across time, with additional (hyper)edges representing gate packing possibilities in the case of GCP\nobreakdash-E.
In both GCP-S and GCP-E, a genetic algorithm is used to generate assignments for each qubit at each time step.
With this sequence of assignments, each state-like edge crossing a partition corresponds to a single state teleportation, while each gate-like edge crossing a partition corresponds to a single gate teleportation, and each gate-like hyperedge crossing a certain number of partitions corresponds to that number of packed gate teleportations.
Note that GCP-S and GCP-E use both teledata and telegates, and additionally, GCP-E considers gate packing (but not gate embedding or detached gates).

Fourth, there is~\citeauthor{burt_multilevel_2026}'s \emph{Multilevel Framework} (MLFM) for graph partitioning~\cite{burt_multilevel_2026}.
The circuit representation is the same as in GCP-E.
However, instead of using a genetic algorithm for distributing, the MLFM method uses a specialized hypergraph partitioning method.
The hypergraph partitioning method used here is a novel one based on the Fiduccia-Mattheyses algorithm~\cite{fiduccia_linear-time_1988} extended to hypergraphs and to the multipartite case, but where the constraints and objective function are adapted to the quantum circuit partitioning problem.
Furthermore, this algorithm is augmented to be multilevel, similarly to KaHyPar~\cite{schlag_high-quality_2023}, which decreases computation time.
It is noted that the best performing variant of the MLFM algorithm is with a \emph{recursive} coarsening algorithm that merges pairs of adjacent time steps, called \emph{MLFM-R}~\cite{burt_multilevel_2026}.

\subsection{A Methodology for Calculating the Gate Fidelity of a QFT~Circuit Implementation} \label{sec:baumer-methodology}

\textcite{baumer_quantum_2024-1} present a methodology for estimating the process fidelity of the QFT+Measurement circuit.
It samples $n$ bit numbers $x := x_{n-1} \dots x_1 x_0$ as input, runs a compute-uncompute circuit on input $\ket{x}$, and measures the result to obtain an $n$ bit number $y$.
The compute-uncompute circuit in question is $\widetilde{\QFTM}_n \of \QFT_n^\dagger$ where $\QFT_n^\dagger$ is a near-ideal implementation of the $n$-qubit inverse QFT and $\widetilde{\QFTM}_n$ is a noisy implementation of the QFT+Measurement circuit.
After sampling and computing probabilities $\prob(y=x)$, an estimator for the process fidelity is computed, where the process fidelity is defined as $F = \left(\frac{1}{2^n}\sum_{x=0}^{2^n-1} \prob(y = x) \right)^2$.

The key aspect of this methodology is that the implementation of $\QFT_n^\dagger \ket{x}$ can be done efficiently and near-ideally by noting that
\begin{align} 
    \QFT_n^\dagger \ket{x_{n-1}\dots x_1 x_0} &= \left( \P(-2 \pi (0.x_0)) \H \ket{0} \right) \nonumber \\
        &\phantom{=} \tensor \left( \P(-2 \pi (0.x_1 x_0)) \H \ket{0} \right) \nonumber \\
        &\phantom{=} \tensor \dots \nonumber \\
        &\phantom{=} \tensor \left( \P(-2\pi (0.x_{n-1}\dots x_0) \H \ket{0} \right)
\end{align}
where the phase gates $\P$ can be implemented as virtual $\RZ$ gates on IBM's superconducting hardware~\cite{mckay_efficient_2017}, making them essentially ideal.
The only source of error is from the single layer of $\H$ gates, which can be done with low error.
This methodology is used to measure the fidelity of the semiclassical QFT circuit with some error mitigation schemes, but the methodology is not applied to other methods of distributing the QFT circuit.

\section{A Distributed QFT Circuit Using Gate Packing} \label{sec:content-proposed-dqft-packed}

\subsection{Definition} \label{sec:proposed-dqft}

We describe a method to distribute the QFT circuit on $n$ qubits across $m$ QPUs of arbitrary sizes $n_1 \dots, n_m$.
Recall the QFT circuit as it is traditionally defined, such as in~\cite{nielsen_quantum_2010} or~\cite{yimsiriwattana_generalized_2004}, which is diagrammed in Fig.~\ref{fig:traditional-qft}.

A key fact to note about $\CP$ gates is that they are symmetric: $\CP(\theta) \ket{\psi}\ket{\phi} = \CP(\theta) \ket{\phi}\ket{\psi}$. 
Therefore, an equivalent diagram of the QFT circuit shown in Fig.~\ref{fig:traditional-qft} can be seen in Fig.~\ref{fig:flipped-qft}, where we flip the role of the control and target qubits for all $\RZ(\theta)$ gates.
\begin{figure*}[!t]
    \centering
    \includegraphics[width=0.7\textwidth]{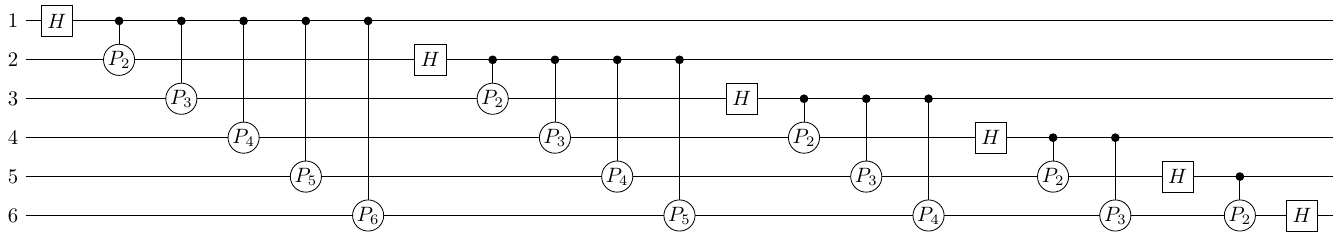}
    \caption[Alternate View of the Traditional QFT Circuit]{Alternate view of the QFT circuit, where the designation of \emph{target} and \emph{control} qubits of the $\CP$ gates are flipped.}
    \label{fig:flipped-qft}
\end{figure*}
This is the same technique that has been used to implement the semi-classical QFT, as was seen in Section~\ref{sec:other-distributed-qft}.

Now, a naive telegate-based circuit distribution scheme would still give the same e-bit requirements as~\cite{yimsiriwattana_generalized_2004}. 
However, we can use the gate packing concept introduced in~\cite{wu_entanglement-efficient_2023} to reduce the number of e-bits required.
Gate packing can be applied to any $\CP$ gates rooted on the same qubit and whose target is the same QPU.
To make this concrete, we can divide the QFT circuit into $n$ distinct time slices, where time slice $j$ is defined as 
\begin{equation}
    T_j := \CP_{j,n} \cdots \CP_{j,j+2} \CP_{j,j+1} \H_j \id_{j-1} \cdots \id_{2} \id_1 \;,
\end{equation}
where the gate subscripts indicate which qubits are acted on (and the angles are omitted).
In other words, slice $j$ consists of a Hadamard gate on qubit $j$ followed by $\CP$ gates rooted on qubit $j$ targeting all qubits $l>j$.
We designate qubit $j$ as the \emph{root qubit} for time slice $j$.
As an example, for the 6\nobreakdash-qubit QFT diagrammed in Fig.~\ref{fig:flipped-qft}, the time slice~1 consists of the first 6 gates ($\H$ and $\CP_2 \dots \CP_6$) and has root qubit~1, while time slice~6 consists of the last gate ($\H$) and has root qubit~6.
Note that with this definition of time slices, the core QFT circuit is traditionally constructed as $\QFT = T_n T_{n-1} \cdots T_2 T_1$.
By the above description, we note that after slice $j$, qubit $j$ has no more gates operating on it: we say it is \emph{inactive}.
So, to check whether a qubit~$j$ is active at time slice~$t$, we check
\begin{equation}
    \act(j) \iff \trace_{\set{1\dots n} \setminus \set{j}}(T_n \cdots T_{t+1}) = \id_j \;,
\end{equation}
where the trace subscript describes a partial trace over all qubits except qubit $j$. 
For the QFT circuit, it can be seen that $\act(j) \iff j \leq t$.
This further leads to the concept of an \emph{active} vs. \emph{inactive} QPU: a QPU $k$ is active if one of its qubits is active, and it is inactive otherwise.
That is, during time slice $t$, 
\begin{equation}
    \act(k) \iff \exists j \in \phi^{-1}(k) : \act(j) \;,
\end{equation}
where $\phi:\set{1 \dots n} \to \set{1 \dots m}$ is the qubit-to-QPU map and $\phi^{-1}$ is its (multivalued) inverse.

Now, for each time slice, we distribute the root qubit's information across all \emph{active} QPUs using starting processes, implement all non-local $\CP$ operations between the root QPU and the other active QPUs using gate packing as the $\CP$ operations share a control for each QPU, and reset the root qubit with ending processes at the end of the time slice.
We can see an example of this circuit for the 6-qubit QFT in Fig.~\ref{fig:distributed-qft}, where QPU~1 has qubit~1, QPU~2 has qubits~$\set{2,3}$, and QPU~3 has qubits~$\set{4,5,6}$.
\begin{figure*}[!t]
    \centering
    \includegraphics[width=0.7\textwidth]{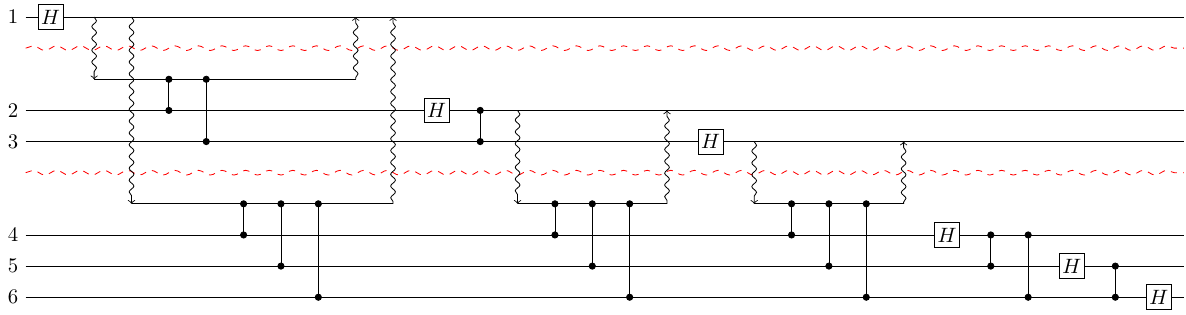}
    \caption[Gate packing Partitioned QFT Circuit Diagram]{Proposed circuit distributed for QFT that includes gate packing. Note that for simplicity, the $\CP$ gate angles are not displayed.}
    \label{fig:distributed-qft}
\end{figure*}

\subsection{Analysis of E-bit Requirements for Optimal QPU Mapping} \label{sec:analysis-new}

We start by analyzing the e-bit requirement for the gate-packed distributed QFT of Section~\ref{sec:proposed-dqft}.
Note that this proposed gate packing scheme follows the given gate ordering of the traditional QFT circuit (compare Fig.~\ref{fig:flipped-qft} with Fig.~\ref{fig:distributed-qft}).
The e-bit requirements for the circuit distribution method we proposed vary depending on the mapping of qubits to QPU $\phi : \set{1 \dots n} \to \set{1 \dots m}$ and on the available QPU sizes $n_1 \dots n_m$, but nothing else.

To show the minimal e-bit requirement according to our gate-packed QFT circuit distribution method, we need to first show a general e-bit formula for any qubit mapping under the constraints of given QPU sizes, after which we can minimize the qubit requirement given QPU sizes.

For each time slice $j$, its root qubit must distribute its state to all other active QPUs, with 1 e-bit needed per (non-root) active QPU. 
Let $a_j$ be the number of active QPUs in slice $j$, including the QPU of the root qubit. 
Then, the number of e-bits required to implement this circuit time slice non-locally is $a_j - 1$, as the root qubit's QPU is the only QPU not needing an e-bit.
Therefore, the total number of e-bits required in this distributed circuit is 
\begin{equation} \label{eq:e-bit formula active}
    E := \sum_{j=1}^{n} (a_j - 1) \;.
\end{equation}
Note that by necessity, $a_1 = m$ (all QPUs start as active) and $a_n = 1$ (all QPUs eventually become inactive, except the one containing the last qubit).
Let $\phi$ be the qubit assignment map (from logical qubit to its assigned QPU) and let $\phi^{-1}$ be its (multivalued) inverse (from QPU to its set of assigned logical qubits).
A given QPU $k$ is active until its last assigned qubit ($\max\phi^{-1}(k)$) becomes inactive: it is active for all time slices $1, \dots, \max\phi^{-1}(k)$.

We can find the optimal qubit assignments in this circuit distribution scheme:

\begin{theorem}
    Let $QPU_k$ be the $k$th smallest QPU. 
    The optimal assignment is for $QPU_1$ to be assigned the first (most significant) $n_1$ qubits, $QPU_2$ the next $n_2$ most significant qubits, etc., until $QPU_m$ is assigned the $n_m$ least significant qubits.
    This can be equivalently stated as QPU $k$ being assigned qubits $$\phi^{-1}(k) = \set{n-\sum_{j=m}^{k+1}n_{j}, \dots, n - \sum_{j=m}^{k}n_{j}+1} \;. $$
\end{theorem}
The proof of this theorem is in the appendix.

Using this result with Equation~\eqref{eq:e-bit formula active}, we get an e-bit count 
\begin{equation} \label{eq:optimal-e-bits}
    E = \sum_{k=1}^m n_k(m-k) \;,
\end{equation}
where $n_k$ are in order of smallest to largest.
Additionally, if we relax this QPU mapping to consider arbitrary QPU orders but still mapping qubits to QPUs sequentially, we obtain the same formula~\eqref{eq:optimal-e-bits}.

Importantly, for the case $m=2$, we get an important result in coherent communication complexity. 
The above formula gives $E = n_1$ where $n_1$ is chosen to be the smallest, tightening the upper bound on the coherent communication complexity of the core QFT and proving:
\begin{theorem}
    The coherent communication complexity of the core QFT (QFT without ending SWAPs) is 
    \begin{equation*}
        Q_0(\QFT_{n_1+n_2}) = \min(n_1,n_2) \;.
    \end{equation*}
\end{theorem}

On a final note, we can state that for evenly-sized QPUs $n_1 = \dots = n_m = \frac{n}{m}$, the gate-packed QFT e-bit requirements formula~\eqref{eq:optimal-e-bits} simplifies to
\begin{equation} \label{eq:optimal-e-bits-even}
    E = \frac{1}{2}n(m-1) \;.
\end{equation}

\subsection{Lower Bound on Any Gate-Packing-Based Distributed QFT}

The traditional QFT circuit, shown in Fig.~\ref{fig:traditional-qft}, has a very specific ordering of gates.
It is this ordering of gates that the distributed QFT algorithm of Section~\ref{sec:proposed-dqft} exploits to get a minimum e-bit count.
However, there are many equivalent circuits which can be obtained by commuting gates: some $\CZ$ and $\H$ gates are defined on different qubits and so trivially commute, and any two $\CP$ gates also commute.
Furthermore, other basic gate-product equivalencies can be established, for example, $(\id \tensor \H) \CRZ(\theta) = \CRX(\theta) (\id \tensor \H)$,%
\footnote{Note $\CRZ$ is closely related to $\CP$, differing only by a phase on the controlled unitary $\RZ$}
 expanding the number of equivalent circuits which can represent the QFT.
Note that including ancilla qubits in our circuit would again expand the number of equivalent circuits for the QFT, but this is not considered in the present work.

Now, assume we have some circuit on $n$ qubits, equivalent to the $n$-qubit core QFT, obtained by applying commutation rules and gate-product equivalencies.
Note that while the specific gates in the circuit may change drastically, these equivalencies do not change the qubit interactions: single-qubit gates remain single-qubit gates and two-qubit gates remain two-qubit gates.%
\footnote{We ignore the trivial possibility of a two-qubit gate being followed by its inverse.}
Therefore, we can find a lower bound for the e-bits necessary to implement the distributed QFT by abstracting away the specific gates in the circuit and only considering the interactions due to two-qubit gates.
This concept gives rise to the definition of an \emph{interaction graph} for a given circuit: a graph such that vertices $V=Q$ are the qubits, and two vertices $(i, j)$ are connected by an edge iff there is a two-qubit gate between them in the circuit.
For our purposes, since we only care about non-local gates, we can define the \emph{non-local interaction graph} as the subgraph of the interaction graph obtained by removing edges corresponding to local gates.
In other words, the non-local interaction graph is such that vertices are qubits and two vertices $i,j$ share an edge iff there is a non-local two-qubit gate between them.
By inspecting the traditional QFT circuit seen in Fig.~\ref{fig:traditional-qft}, we note the interaction graph of the QFT circuit is given by the complete graph $K_n$; for distributing into QPUs of sizes $n_1, \dots, n_m$, the non-local interaction graph is given by the complete $m$-partite graph $K_{n_1,\dots,n_m}$.
We see examples of such graphs in Fig.~\ref{fig:interaction-graphs}.

To distribute such a circuit using gate teleportations and packed gate teleportations, we need to decide for each two-qubit gate which qubit is considered the root qubit and which is considered the target qubit.
It is a necessary condition for packing that two gates have the same root qubit.

\begin{figure}[h]
    \centering
    \subfloat[]{\includegraphics[width=0.4\linewidth]{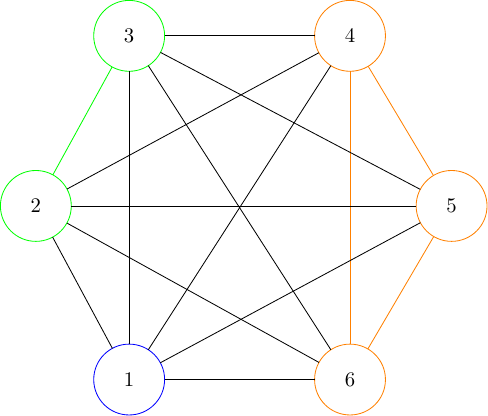} \label{fig:connectivity}}
    \hfil
    \subfloat[]{\includegraphics[width=0.4\linewidth]{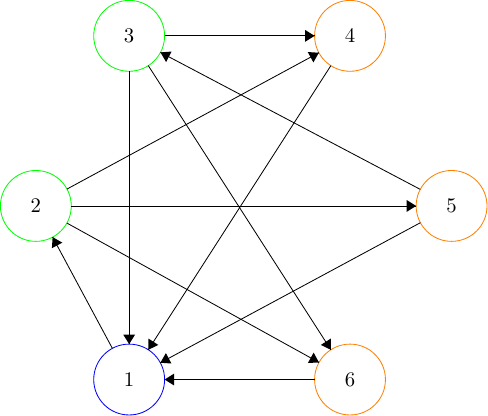} \label{fig:non-local-connectivity-directed}}
    \\
    \subfloat[]{\includegraphics[width=0.8\linewidth]{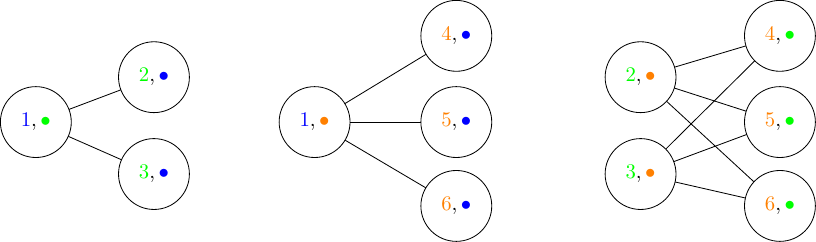} \label{fig:constraints}}
    \caption{
        (a) $K_6$, the interaction graph for a 6-qubit QFT. Node colours correspond to QPU assignment, but are irrelevant for this (local) interaction graph. 
        (b) $K_{1,2,3}$, the non-local interaction graph for a 6-qubit distributed QFT on QPUs of sizes $1,2,3$. 
        An arbitrary assignment of root and target qubits is given by edge directions. 
        (c) $K_{1,2} \cup K_{1,3} \cup K_{2,3}$, the constraint graph for a 6-qubit distributed QFT on QPUs of sizes $1,2,3$.
        Note we identify the QPUs by their colour here (as opposed to the main text, where we use their indices $k$). 
    }
    \label{fig:interaction-graphs}
\end{figure}

We can describe an arbitrary choice of root and target for each of the edges in the QFT's non-local interaction graph by making it a directed graph, where edges are directed from root qubit to target qubit.
Now, for a given directed non-local interaction graph, we need to determine how many e-bits are necessary for a given choice of roots and QPU assignment. 
Any group of gates with the same root qubit and targeting qubits on a common QPU will only need a single e-bit to distribute due to gate packing. 
In contrast, if there is no gate rooted at $i$ targeting some QPU $k$, then no e-bits are needed for that case.
To determine the minimum number of e-bits required given a non-local interaction graph, we must maximize the pairs of qubit $i$ and QPU $k$ for which there are no non-local gates from qubit $i$ to a qubit in $k$.
 
Enforcing no non-local gate rooted at qubit $i$ and targeting QPU $k$ for the QFT's non-local interaction graph means letting all edges between $i$ and each $j \in \phi^{-1}(k)$ target $i$.
However, by symmetry, this implies that there are non-local gates rooted on each $j$ and targeting QPU $\phi(i)$.
So, each choice of qubit $i$ and QPU $k$ for which we want to enforce no non-local gate from $i$ to $k$ enforces a constraint on the choice of other qubits and QPUs for which we cannot do the same.
For example, for the root-target choices of Fig.~\ref{fig:non-local-connectivity-directed}, we have that there is no e-bit needed to distribute any gates from the root qubit~1 to the orange QPU.
However, this implies that each qubit of the orange QPU has a gate targeting the blue QPU, and so an e-bit is necessary to distribute each of those three non-local gates.

The set of constraints described above can be encoded into a graph such that vertices are root-qubit-to-target-QPU pairs and there is an edge between $(i,k)$ and $(i',k')$ iff we cannot enforce both no non-local gates between $i$ and $k$ and no non-local gates between $i'$ and $k'$.
This happens when root $i$ is among the qubits $\phi^{-1}(k')$ of target QPU $k'$ and when $i'$ is in $\phi^{-1}(k)$.
As an example in Fig~\ref{fig:non-local-connectivity-directed}, choosing no e-bits to distribute from root qubit~1 targeting the orange QPU implies that there must be an e-bit to distribute from root qubit~4 targeting the blue QPU, i.e., the constraint graph contains an edge between $(1,\text{orange})$ and $(4,\text{blue})$.
The same constraints hold for qubits~4 and~5.
Applying the idea of a constraint graph to the QFT circuit, this gives a graph with components
\begin{equation} \label{eq:qft-constraints}
    K_{n_{k},n_{k'}} \;\text{for all distinct QPU pairs $k \neq k'$} \;,
\end{equation}
where $K_{n_{k},\,n_{k'}}$ are the complete bipartite graphs with one part having $n_k$ vertices and the other part having $n_{k'}$ vertices. 
An example can be seen in Fig.~\ref{fig:constraints}.

In this context, an optimal circuit distribution is a maximum independent set in the constraint graph: each vertex in the maximal independent set is a choice of qubit-root-to-QPU-target pair that won't use an e-bit, where independence implies that the constraint is satisfied.
Dually, we have a minimum vertex cover: each qubit-root-to-QPU-target vertex in the cover is a packing process, and the edges each covers are the two-qubit gates packed that use that e-bit.
Therefore, the size of the minimum vertex cover is exactly the minimum number of e-bits necessary to partition the circuit, according to its connectivity.
Looking back to the constraint graph for QFT, highlighted in Equation~\eqref{eq:qft-constraints}, the minimum vertex cover is  
\begin{equation} \label{eq:qft-vertex-cover}
    \sum_{k \neq k'} \min(n_k, n_k')
\end{equation}
as each component is independent.
This can easily be proven using the relation between the vertex cover number and the independence number~\cite{weisstein_independence_nodate}.
Therefore, when the QPUs are indexed by size (without loss of generality), the number of e-bits required for the QFT implemented by gate packing is lower bounded by
\begin{equation}
    \sum_{k \neq k'} \min(n_k, n_k') = \sum_{k=1}^{m}\sum_{k'=k+1}^{m} n_k = \sum_{k=1}^{m-1} n_k (m-k) \;,
\end{equation}
which is tight as it matches the e-bit count of our method as stated in Equation~\eqref{eq:optimal-e-bits}.%
\footnote{Remark: this proof method was independently conceived, but bears a striking similarity to the circuit partitioning algorithm of \citeauthor{g_sundaram_efficient_2021}~\cite{g_sundaram_efficient_2021}, also used in the vertex partitioning refinement step of the Pytket-DQC workflows~\cite{andres-martinez_distributing_2024} as described in Section~\ref{sec:general-circuit-partitioners}.
However, this proof method is done by generalizing over the gate types to establish a lower bound, while the partitioning algorithm considers the specific set of gates used in the circuit to partition and produces a valid partitioned circuit.}

\section{A Distributed QFT Circuit Using Detached Gates} \label{sec:content-proposed-dqft-detached}

The previous gate-packed distributed QFT does not use detached gates in its scheme, which ensures only a single communication ancilla per QPU is necessary at the cost of a higher lower bound on e-bit cost when at least 3 QPUs are in use.

Here, we present a novel scheme for the distributed QFT that maximally takes advantage of gate packing with detached gates.
Instead of directly giving a circuit construction, we present a minimal-cost hypergraph partition as in the framework of~\cite{andres-martinez_distributing_2024}, which can be straightforwardly extracted into an implementable circuit using standard methods.

\subsection{QFT Hypergraph-Based Distribution} \label{sec:proposed-dqft-2}

We can describe the circuit for the $n$-qubit QFT using the following parts
\begin{align*}
    Q &= \set{q_1 \dots q_n} \\
    G &= \set{g_{i,j} \;:\; 1 \leq i < j \leq n}
\end{align*}
which label the $n$~qubits and $n(n-1)/2$~two-qubit gates.

Now, we note that in the traditional QFT circuit, each $H$ gate splits a qubit's two-qubit gates between those before $H$, which can all be packed together, and those after $H$, which can also all be packed together.
However, no packing is allowed across the $H$ gate with that qubit as root~\cite{wu_entanglement-efficient_2023}.
For our hypergraph, the nodes are given by $Q \cup G$, the qubits and two-qubit gates. 
The set of packets consists of two packets for each qubit, with one packet before the $H$ gate and one packet after the $H$ gate, with the exception of qubits $1$ and $n$ which have only a post-$H$ and pre-$H$ packet respectively.
We label the pre-$H$ packets as right-packets $R$, and post-$H$ packets as left-packets $L$.
Summarizing, the packets are given by $P = L \cup R$, where
\begin{align*}
    L &= \set{l_1 \dots l_{n-1}} \\
    R &= \set{r_2 \dots r_{n}} \,,
\end{align*}
and where 
\begin{align*}
    l_i &= \set{ q_i } \cup \set{ g_{i,j} \;:\; j > i} \\
    r_j &= \set{ q_j } \cup \set{ g_{i,j} \;:\; i < j} \,.
\end{align*}
Note that each gate $g_{i,j} \in G$ is present in exactly two packets: $g_{i,j} \in l_i, r_j$.
Due to this fact, we can represent this hypergraph instead as a bipartite graph, with each packet as a node ($L$ nodes on the left and $R$ nodes on the right) and each two-qubit gate as an edge connecting the packets it is a part of.

A choice of distribution is a map $\phi: Q \cup G \to \set{1 \dots m}$ from the qubits and gates to the $m$ QPUs. 
Note that the qubit-to-QPU part of $\phi$ corresponds to the qubit layout, as in previous sections, while the gate-to-QPU assignment is new and corresponds to where the gate is implemented.
If $ \phi(g_{i,j}) = \phi(q_i) = \phi(q_j)$, then the gate is local; if $\phi(g_{i,j}) = \phi(q_i) \neq \phi(q_j)$, then a telegate is used with the starting/ending processes from $q_j$ to QPU $\phi(g_{i,j})$; if $\phi(g_{i,j}) = \phi(q_j) \neq \phi(q_i)$, then a telegate is used with the starting/ending processes from $q_i$ to QPU $\phi(g_{i,j})$; if $\phi(q_i) \neq \phi(g_{i,j}) \neq \phi(q_j)$, then a detached gate is used with both qubits $q_i$ and $q_j$ being distributed to QPU $\phi(g_{i,j})$.
The connectivity of a packet $p$ is
\begin{align*}
     \lambda(p) &= \abs{\set{\phi(a) \;:\; a \in p}} \\
     &= \abs{\set{\phi(g) \;:\; g \in p \cap G} \cup \set{\phi(q)}} \;,
\end{align*}
i.e. the number of QPUs that $p$ is partially assigned to.
This corresponds to an e-bit cost of 
\begin{align}
    E(p) &= \lambda(p)-1 \nonumber \\
    &= \abs{\set{\phi(g) \;:\; g \in p \cap G} \setminus \set{\phi(q)}} \;, \label{eq:packet-cost}
\end{align}
as the root qubit's QPU does not need to be distributed to.
Note that to be valid given a set of QPU sizes $n_1 \dots n_m$, we need $\abs{\phi^{-1}(k) \cap Q} \leq n_k$.
This is a restriction on computational qubits; for our detached-gate-based distributed QFT, we consider the communication ancillas as unrestricted.

To translate between $\phi$ and the bipartite packet-gate graph representation, we consider each pair of nodes $l_q$ and $r_q$ to be assigned $\phi(q)$.
Therefore, an assignment in bipartite representation is a colouring of edges and nodes, where node $r_q$ on the right is forced to be coloured the same as the node $l_q$ on the left. 

We will show the cost $E = \sum_{p \in P} (\lambda(p)-1)$ of the entire assignment can be minimized by choosing as many packets as possible having cost 0, then as many as possible having cost 1, then the rest having cost 2.
These assignments will involve a hub QPU, as we will detail later.

We will do this by assigning each qubit $q$ the packets costs $(c_q^l, c_q^r)$, where $c_q^l = \lambda(l_q)-1$ is the cost of the $l$-packet of qubit $q$ (and similarly for $r$).

Note all the following lemmas and proofs are equally valid when interchanging right and left and interchanging less-than and greater-than.
This is due to the QFT's symmetry $i \mapsto n-i+1$.

\begin{lemma}
    Let $q,q'$ be on different QPUs, i.e. $\phi(q) \neq \phi(q')$. 
    Let $q$ have cost $(0, c_r)$ and $q'$ have cost $(c_l, 0)$.
    Then $q > q'$.
\end{lemma}
\begin{proof}
    We will prove the contrapositive.
    Assume $q < q'$, so there is a gate $g_{q,q'}$.
    Since $c^l_q=0$, we know $\phi(g_{q,q'}) = \phi(q)$. 
    Similarly, since $c^r_{q'}=0$, we know $\phi(g_{q,q'}) = \phi(q')$.
    By transitivity, $\phi(q)=\phi(q')$.
\end{proof}

\begin{corollary} \label{cor:hub-uniqueness}
    Let $q,q'$ both be (0,0)-cost qubits. Then $\phi(q)=\phi(q')$.
\end{corollary}
\begin{proof}
    If the two qubits were on different QPUs, by the previous lemma and its symmetric counterpart we would have $q < q'$ and $q > q'$, which is a contradiction.
\end{proof}

Note this motivates the \emph{hub} terminology: a circuit can have at most one QPU with cost-(0,0) qubits, and by definition, all gates to these qubits are implemented on that QPU, making that QPU a home for many gates.

\begin{lemma}
    Consider a qubit $q'$, and let $S$ be the set of previous qubits with 0 left-cost, i.e. $S = \set{q < q' \;:\; c^l_q = 0}$.
    Then $c^r_{q'} \geq \abs{\phi(S) \setminus \set{\phi(q')}}$, i.e. the right-cost is at least the number of other QPUs with a 0 left-cost qubit preceding $q'$.
\end{lemma}
\begin{proof}
    All gates $g_{q,q'}$ for $q \in S$ connect to $r_{q'}$.
    Furthermore, since $c^l_q=0$, we know $\phi(g_{q,q'}) = \phi(q)$.
    Finally, substituting $r_{q'}$ into~\eqref{eq:packet-cost}, we get
    \begin{align*}
        c^r_{q'} &= \abs{\set{\phi(g_{q,q'}) \;:\; q < q'} \setminus \set{\phi(q')}} \\
        &\geq \abs{\set{\phi(g_{q,q'}) \;:\; q < q' \;:\; c^l_q = 0} \setminus \set{\phi(q')}} \\
        &= \abs{\phi(S) \setminus \set{\phi(q')}}
    \end{align*}
    by definition.
\end{proof}

\begin{corollary} \label{cor:01-quasi-uniqueness}
    At most two QPUs have a qubit with cost $(0,1)$ (and similarly for cost $(1,0)$).
\end{corollary}

\begin{corollary} \label{cor:01-uniqueness}
    If there is a QPU with a cost $(0,0)$ qubit, then at most one other QPU has a cost $(0,1)$ qubit (and similarly for $(1,0)$).
\end{corollary}

\begin{corollary}
    Qubits which have some left-cost of~0 or right-cost of~0 must be ordered as follows:
    $$(2,0) < (1,0) \leq (1,0) < (0,0) < (0,1) \leq (0,1) < (0,2)$$
    Note the $\leq$ since there may be two distinct QPUs with qubits of cost $(0,1)$ (and similarly for $(1,0)$), but these qubits can be in any order relative to each other (i.e. the two such QPUs can be interleaved).
\end{corollary}
\begin{corollary}
    Qubits which have left- and right-cost of at most~1 must be ordered as follows:
    \begin{align*}
        (1,0) < (1,1) \leq \dots &\leq (1,1) \\
        &\;\;\; \leq (0,0) \leq \\
        &\;\;\;\;\;\;\;\;\;\; (1,1) \leq \dots \leq (1,1) < (0,1) \;.
    \end{align*}
    Note there may be arbitrarily many QPUs with qubits of cost $(1,1)$, and they may be in any order (possibly interleaved) with respect to each other and to the qubits of cost $(0,0)$.
\end{corollary}

We are now ready to find the minimal cost QFT hypergraph partition. 
Let $n_1 \leq \dots \leq n_m$ be QPU sizes.

\begin{lemma}[Hub Model Bound] \label{lemma:hub-model-bound}
    For an assignment containing a qubit of cost $(0,0)$, the optimal cost is bounded below by $2 \sum_{k=1}^{m-3} n_k + n_{m-2} + n_{m-1}$.
\end{lemma}
\begin{proof}
    An optimal assignment will have as many qubits as possible with cost 0, then as many as possible with cost 1, then as many as possible with cost 2, etc.
    Let QPU~$m$ (the largest QPU) have all its qubits at cost $(0,0)$, giving this QPU a total e-bit cost of 0.
    Now, by Corollary~\ref{cor:hub-uniqueness}, there can be no more qubits with cost 0.
    Let QPUs $m-1$ and $m-2$ (the next two largest QPUs) have all their qubits at cost $(0,1)$ and $(1,0)$ respectively. 
    This means these two QPUs have a total e-bit cost of $n_{m-1}$ and $n_{m-2}$ respectively.
    Now, by Corollary~\ref{cor:01-uniqueness}, there can be no more qubits with cost 1.
    Let all remaining qubits have cost 2. 
    This means that the remaining QPUs have a total e-bit cost of $2n_k$ for $k \in [1 \dots m-3]$.
\end{proof}

\begin{lemma} \label{lemma:hub-model-achievable}
    The e-bit lower bound of Lemma~\ref{lemma:hub-model-bound} is achievable.
\end{lemma}
\begin{proof}
    Let the first $n_{m-2}$ qubits be assigned to QPU $m-2$, the next $n_m$ qubits assigned to QPU $m$, the last $n_{m-1}$ qubits assigned to QPU $m-1$, and the remaining qubits be assigned to the remaining QPUs.
    Let all gates be either local (if both its qubits are in the same QPU) or assigned to the hub QPU $m$.
    This clearly respects $c^l_q,c^r_q \leq 1$ for all qubits $q$ since all gates have at most two QPU choices (one of which is the same QPU as the packet root).
    
    Furthermore, this clearly respects the $(0,0)$ cost of QPU $m$, since its non-local gates are assigned to the same QPU as its qubits (i.e. implemented locally).
    Now for the qubits on the QPU $m-2$ (the first QPU), each left packet has a gate to a qubit in the hub, while each right packet only has qubits to the same QPU.
    This implies the left packets has some non-local gate, but the right packets are all local gates, and therefore the $(1,0)$ cost is respected.
    By symmetry, the $(0,1)$ cost is respected for the qubits on QPU $m-1$ (the last QPU).
    
    Finally, all other qubits have a qubit of another QPU both before and after them, requiring a non-local gate both in the left and right packets. 
    This ensures the $(1,1)$ cost is respected. 
\end{proof}

\begin{lemma}[Hubless Model Bound]
    For an assignment containing no qubit of cost $(0,0)$, the optimal cost is at least $E \geq 2 \sum_{k=1}^{m-4} n_k + n_{m-3} + n_{m-2} + n_{m-1} + n_m$.
\end{lemma}
\begin{proof}
    An optimal assignment will have as many qubits as possible with cost 0 (none, by the lemma's assumption), then as many qubits as possible with cost 1, then as many as possible with cost 2, etc.
    The cost~1 qubits must have cost~$(0,1)$ or~$(1,0)$.
    By Corollary~\ref{cor:01-quasi-uniqueness}, at most two QPUs can have qubits of cost~$(0,1)$, and at most two QPUs can have qubits of cost~$(1,0)$, meaning at most four QPUs can contain cost~1 qubits.
    To maximize the number of cost~1 qubits, these must be the four largest QPUs, contributing a total e-bit cost of $n_{m}+n_{m-1}+n_{m-2}+n_{m-3}$.
    The remaining qubits must have cost at least 2, meaning each remaining QPU has an e-bit cost of at least $2n_k$ for $k \in [1 \dots m-4]$.
\end{proof}

\begin{lemma}
    The hubless model is never better than the hub model.
\end{lemma}
\begin{proof}
    $E_\text{hubless}-E_\text{hub} = n_m - n_{m-3} \geq 0$ since $n_k$ are ordered.
\end{proof}

\begin{theorem}
    For the traditional (core) QFT circuit, the minimal e-bit count across any assignment of qubits to QPUs of sizes $n_1 \leq \dots \leq n_m$ (excluding any communication ancillas) is exactly 
    \begin{equation} \label{eq:packed-cost}
        E = 2\sum_{k=1}^{m-3} n_k + n_{m-2} + n_{m-1} \;.
    \end{equation}
\end{theorem}
\begin{proof}
    We have shown the hub model is achievable and optimal given the existence of qubits of cost $(0,0)$. 
    We have shown this cannot be improved by excluding qubits of cost~$(0,0)$.
    Corollaries~\ref{cor:hub-uniqueness},~\ref{cor:01-quasi-uniqueness}, and~\ref{cor:01-uniqueness} imply there are no other ways to decrease the number of qubits with cost at least 2.
\end{proof}

Note that for fewer than four QPUs, we ignore the smallest QPU size terms in Equation~\eqref{eq:packed-cost}, giving $E= n_1 + n_2$ for three QPUs and $E=n_1$ for two QPUs.

For the specific case of equal-sized QPUs (\mbox{$n_1 = \dots = n_m = \frac{n}{m}$}), Equation~\eqref{eq:packed-cost} simplifies to
\begin{align*} \label{eq:detached-count-equal}
    E &= \begin{cases}
        \frac{2}{m}n(m-2) & m \geq 3 \\
        \frac{1}{2}n & m=2
    \end{cases} 
\end{align*}
where we note the case $m=2$ matches our gate-packed distributed QFT case.

\section{Analysis of e-bit Requirements of Previous Distributed QFTs} \label{sec:analysis-previous}

The previously established partitions of the QFT circuit, described in Section~\ref{sec:other-distributed-qft}, are typically analyzed only in terms of their asymptotic complexity rather than via exact expressions, and are usually described in the setting of equal-sized QPUs.
In this section, we generalize the analysis of e-bit requirements for those previous partitioning schemes and compare the results to the gate-packed QFT.

\subsection{The Naive Distributed QFT} \label{sec:comparison-yim}
The first analytical procedure for specifically distributing the QFT circuit is from~\cite{yimsiriwattana_generalized_2004}, which we call the naive distributed QFT.
To compare its e-bit count with our gate-packed QFT, we need to generalize their analysis, which was done only informally and restricted to equal QPU sizes.
Here we follow a similar argument as~\cite{yimsiriwattana_generalized_2004} but described more formally and applied to the case of distinct QPU sizes $\set{n_1, \dots, n_m}$.

The $n$ qubit QFT circuit consists of $n(n-1)/2$ two-qubit gates~\cite{nielsen_quantum_2010}.
Note from~\cite{cleve_fast_2000,camps_quantum_2021} that the QFT circuit can be decomposed recursively: for a bipartition $A, B$ of the set of qubits $A \cup B = \set{1,\dots,n}$, we can determine that
\begin{equation*}
    \QFT_{A \cup B} = (V_{A \cup B}) (\QFT_B \tensor \id_A) (U_{A \cup B}) (\id \tensor \QFT_A) \;,
\end{equation*}
where $\QFT_A$ and $\QFT_B$ are smaller local QFT operations and $U,V$ are the left-over operations acting between $A$ and $B$ (composed of some $\CP$ gates only).
This decomposition can be applied recursively so that any partition of the $n$ qubits into QPUs $1, \dots, m$ gives rise to a decomposition
\begin{equation*}
    \QFT_{} = U_m (\id \tensor \QFT_m) \dots U_1 (\id \tensor \QFT_1) \;,
\end{equation*}
where $\QFT_k$ is an $n_k$-qubit QFT acting locally on QPU~$k$ and the operators $U_1, \dots, U_m$ are non-local operations composed of $\CP$ gates only.
For such a partition into QPUs, each $\QFT_k$ operates on $n_k$ qubits and therefore consists of $n_k(n_k-1)/2$ two-qubit gates, while the remaining two-qubit gates are between each qubit and all qubits of the other QPUs and therefore consist of the $\sum_{k=1}^{m}n_k \sum_{j=k+1}^m n_j$ leftover two-qubit gates.
Since the naive method uses 1 e-bit per non-local gate, this gives a total e-bit requirement of
\begin{equation} \label{eq:yim-count}
    E = \sum_{k=1}^{m-1}n_k \sum_{j=k+1}^m n_j \;.
\end{equation}
For the specific case of equal-sized quantum processing units (\mbox{$n_1 = \dots = n_m = \frac{n}{m}$}), Equation~\eqref{eq:yim-count} simplifies to
\begin{align*} \label{eq:yim-count-equal}
    E &= \frac{n}{2m}n(m-1) \;.
\end{align*}

\begin{table*}[t]
    \caption{Comparison of QPU interconnect methods and resource overhead.}
    \label{tab:qpu_comparison}
    \centering
    \renewcommand{\arraystretch}{1.3}
    \begin{tabular}{@{}llcc@{}}
        \toprule
        \textbf{Method} & \textbf{E-bits (any size QPUs)} & \textbf{E-bits (equal size QPUs)} & \textbf{Max comm. ancillas per QPU} \\ 
        \midrule
        Naive & $\sum_{k=1}^{m-1} n_k \sum_{j=k+1}^m n_j$ & $\frac{n}{2m} n(m-1)$ & 1 \\
        Linear topology & $\sum_{k=1}^{m-1} n_k \cdot 2(m-k)$ & $\phantom{\frac{n}{2m}}n(m-1)$ & 2 \\
        \textbf{Gate-packing-based (Ours)} & $\mathbf{\sum_{k=1}^{m-1} n_k \cdot (m-k)}$ & $\mathbf{\phantom{\frac{1}{m}}\frac{1}{2} n(m-1)}$ & \textbf{1} \\
        \textbf{Detached-gate-based (Ours)}& $\mathbf{\sum_{k=1}^{m-3} 2n_k + n_{m-2} + n_{m-1}}$ & $\mathbf{\phantom{\frac{1}{2}}\frac{2}{m} n(m-2)}$ & $\mathbf{\leq n-n_m}$\\ 
        \bottomrule
    \end{tabular}
\end{table*}

\subsection{The Linear-Topology-Based Distributed QFT} 

The linear topology distributed QFT~\cite{rodney_van_meter_communications_2004} has only been studied for equal QPU sizes; here, we will proceed with an alternative analysis that can be applied to arbitrary QPU sizes.

First, we examine the QFT as it has been mapped to the linear nearest-neighbour~(LNN) topology without distribution. 
This appears as the base case in~\cite{jin_optimizing_2024}, and such a circuit is shown in Figure~\ref{fig:linear-qft}.
Next, we partition the circuit from bottom to top into QPUs of size $n_1, \dots, n_m$.
Then, we note that the $\CP$ gates can be packed with their adjacent $\SWAP$ gates.
This implies we only need to count the $\SWAP$ gates to know the number of e-bits required.
Since between qubits $j$ and $j+1$ (counting starting with qubit~1 at the bottom) there are $j$ $\SWAP$ gates, the total number of $\SWAP$ gates is 
\begin{equation*}
    \sum_{k=1}^{m-1} \sum_{k'=1}^{k} {n_{k'}} =  \sum_{k=1}^{m-1} {n_{k}} (m-k) \;,
\end{equation*}
which is equal to the e-bit count for the gate-packed QFT partitioning scheme.
Since each non-local $\SWAP$-followed-by-$\CP$ operation needs 2 e-bits to distribute,%
\footnote{The $\SWAP$ gate takes 2~e-bits to distribute. This can be done either with the telegate protocol by using an embedding process (as shown in \cite[Fig.~3]{wu_entanglement-efficient_2023}) while gate packing the $\CP$ gate, or with the teledata protocol by using bidirectional state teleportation while implementing the $\CP$ locally between the two state teleportations (as shown in Fig.~\ref{fig:linear-qft-distributed}).} %
the number of e-bits required is twice that of the gate-packed QFT highlighted in Equation~\eqref{eq:optimal-e-bits}:
\begin{equation} \label{eq:linear-general-e-bits}
    E = 2 \sum_{k=1}^{m-1}  n_k (m-k) \;.
\end{equation}

\subsection{Summary of e-bit Requirements for the Various QFT Partition Schemes}

We summarize the e-bit requirements of Sections~\ref{sec:content-proposed-dqft-packed},~ \ref{sec:content-proposed-dqft-detached}, and~\ref{sec:analysis-previous} for any $m$ QPU sizes $n_1, \dots, n_m$ and for $m$ equal-sized QPUs of sizes $\frac{n}{m}$ in Table~\ref{tab:qpu_comparison}, as well as the extra ancillas needed to hold the e-bit halves during distribution, called \emph{communication ancillas}.
Note that when scaling the distributed quantum system with a constant number of communication ancillas per QPU, our gate-packed distributed QFT has a constant factor improvement of $1/2$ over the next best method (linear topology) and a quadratic improvement over the naive distributed QFT (for a fixed number of QPUs).
However, when allowing an unbounded number of communication ancillas, our detached-gate-based distributed QFT linearly improves the minimal e-bit count for $m>2$.
Furthermore, we note two special cases.
For two QPUs ($m=2$), our gate-packed distributed QFT achieves the same e-bit count as our detached-gate-based distributed QFT, which we know is optimal by the coherent communication complexity bound~\eqref{eq:chen-bound}.
In the special case of $n$ single-qubit QPUs ($m=n$), the naive distributed QFT achieves its best e-bit count, which is equal to the e-bit count for our gate-packed distributed QFT.

\section{Empirically Determined e-bit Counts by Circuit Partitioners} \label{sec:content-empirical}

\begin{figure*}[!t]
    \centering
    \subfloat{\includegraphics[width=0.48\textwidth]{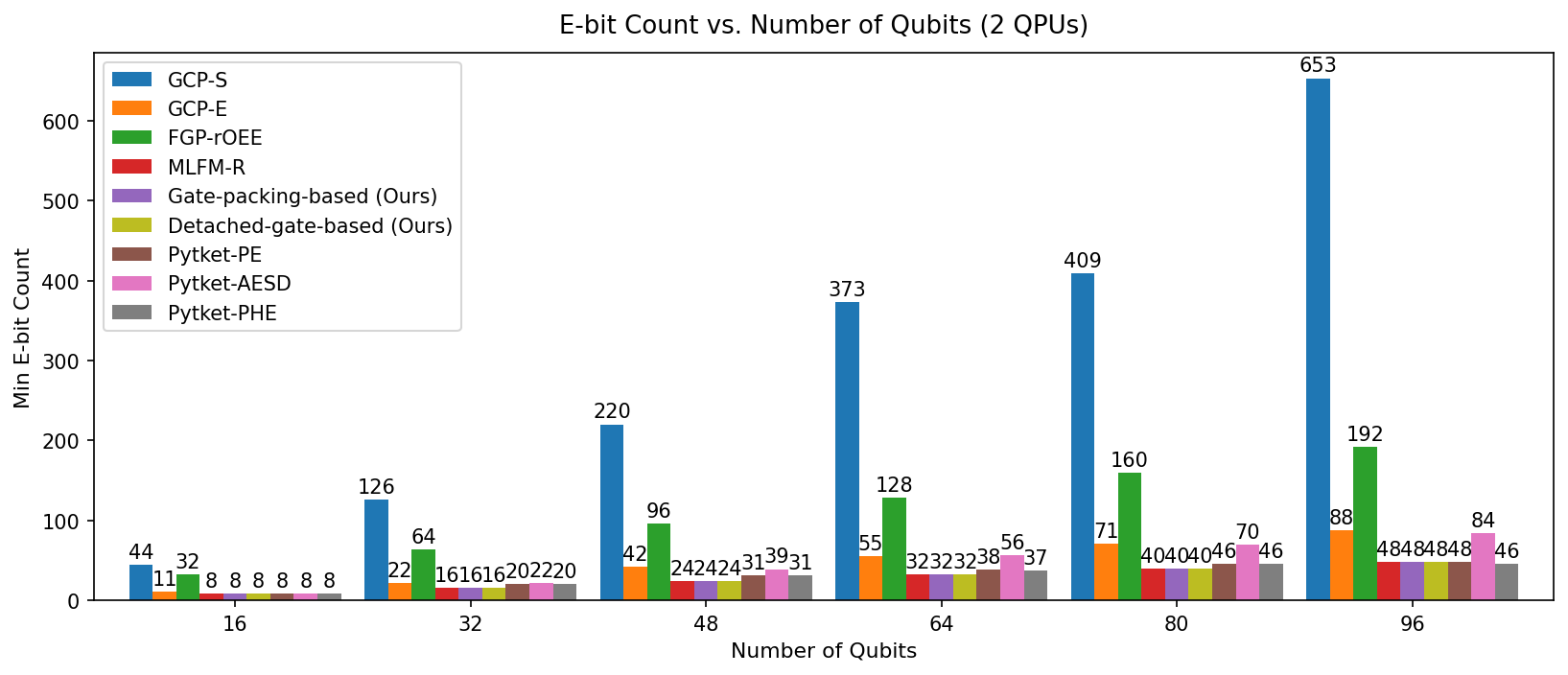}}
    \hfil
    \subfloat{\includegraphics[width=0.48\textwidth]{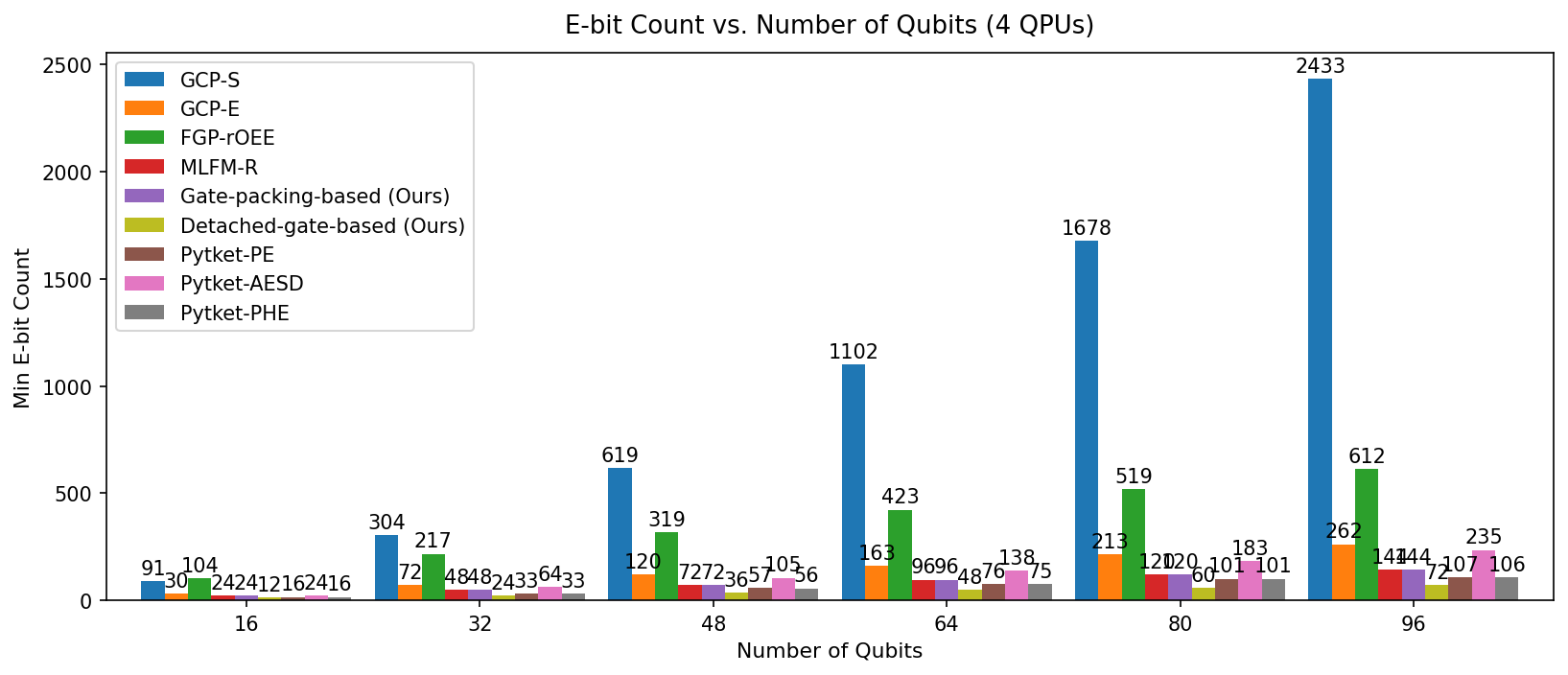}}
    \caption[Circuit Partitioners' Resulting e-bit Counts for QFT Partitioning]{Typical e-bit counts for partitioning various QFT sizes into 2 and 4 QPUs.}%
    \label{fig:few-QPUs-e-bits}
\end{figure*}

\begin{figure*}[!t]
    \centering
    \subfloat{\includegraphics[width=0.48\textwidth]{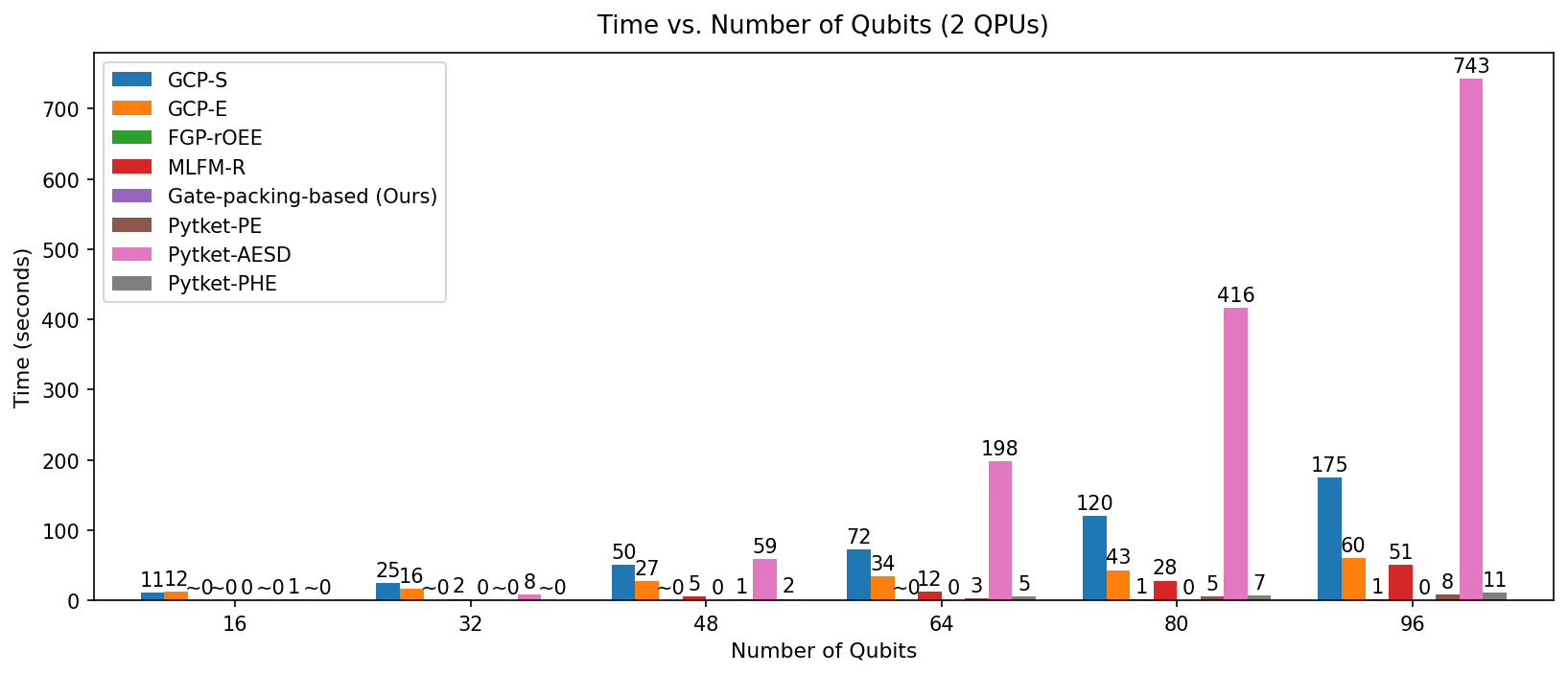}}
    \hfil
    \subfloat{\includegraphics[width=0.48\textwidth]{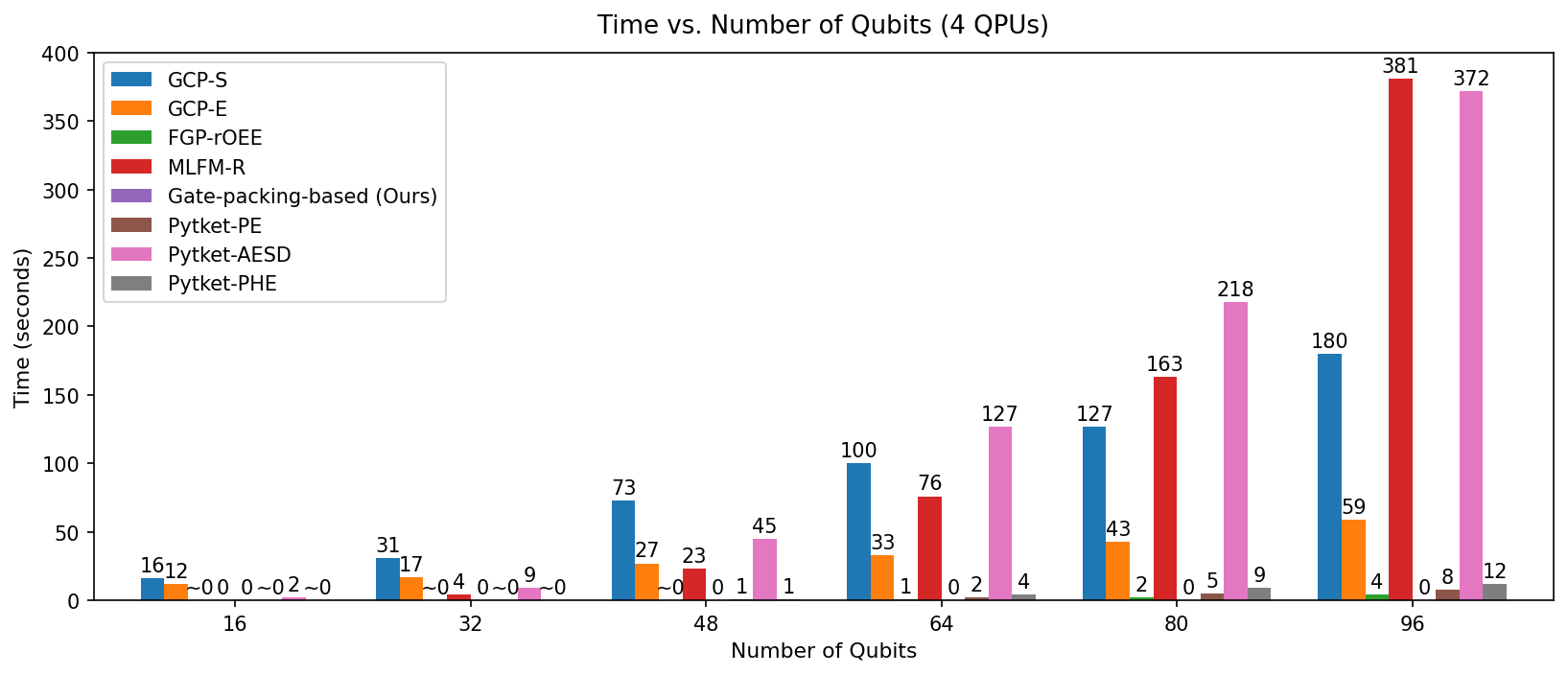}}
    \caption[Circuit Partitioners' Runtime for QFT Partitioning]{Typical time for partitioning various QFT sizes into 2 and 4 QPUs on a laptop. Note that the time for our detached-gate-based distributed QFT is not included, as we did not implement the hypergraph-to-circuit conversion algorithm.}
    \label{fig:few-QPUs-time}
\end{figure*}

In this section, we perform a comparative study of circuit partitioners.
First, we see how the QFT partitioning schemes of Sections~\ref{sec:content-proposed-dqft-packed} and~\ref{sec:content-proposed-dqft-detached} compare with the partitioning results from the circuit partitioners in the literature (see Section~\ref{sec:general-circuit-partitioners}) in terms of e-bits required to implement the partitioned circuit.
Later, we discuss the reasons for the outlier results.

Many of the circuit partitioning algorithms of Section~\ref{sec:general-circuit-partitioners} have already been implemented by \citeauthor{burt_multilevel_2026} and generously shared online in the DisQCO repository~\cite{burt_disqco_2025}.%
\footnote{See online at \href{https://github.com/felix-burt/DISQCO}{github.com/felix-burt/DISQCO}~\cite{burt_disqco_2025}.}
Note that this code repository was still a work-in-progress during the course of this paper.
For the purposes of this paper, the DisQCO repository was forked, the Pytket-DQC workflows were integrated into the code base, and a main.py file was added where all partitioners were run and their results recorded.%
\footnote{See online at \href{https://github.com/ZacharyVernec/DISQCO/tree/zacharyvernec}{github.com/ZacharyVernec/DISQCO/tree/zacharyvernec}~\cite{vernec_fork_2025}.}

\subsection{Methodology} \label{sec:empirical-methodology}

The methodology generally follows \citeauthor{burt_multilevel_2026}~\cite[Fig. 28]{burt_multilevel_2026} but includes more circuit partitioners: all circuit partitioners along with all their variants (for GCP) and all their workflows (for Pytket) described in Section~\ref{sec:general-circuit-partitioners} are included in the methodology. 
The partitioners are run on $\QFT$ circuits of sizes from 16 to 96 (in steps of 16) and configured to partition into either 2 or 4 equal-sized QPUs.
The network topology is assumed to be fully connected.
The metrics measured for each partition run are (1) the number of e-bits required to implement the circuit partition found by the partitioner and (2) the time it takes to initialize and run the circuit partitioner on a laptop.%
\footnote{A laptop with a 12th Gen Intel(R) Core(TM) i5-1235U (1.30 GHz) and 16.0 GB RAM running the code in Ubuntu 22.04 LTS on WSL-2 (Windows Subsystem for Linux).}%
\footnote{Note that the time taken to create the QFT circuit object and the time taken to create a network object representing the distributed quantum system are not included.}

The FGP-rOEE partitioner is run using default arguments.
In particular, the single-qubit gates are not removed from the circuit before partitioning, as would be done if the partitioned circuit were to be saved for later execution.
The GCP variants (GCP-S and GCP-E) are run with the default arguments in DisQCO~\cite{burt_disqco_2025}, which means the genetic algorithm uses a population size of 100 partitions and runs for 100 generations with a mutation rate of 0.9.
The MLFM-R partitioner is also run with default arguments, in particular limiting each partitioning pass to implement at most $n$ state teleportations (where $n$ is the number of qubits in the QFT circuit).
The Pytket workflows~(PE, AESD, and PHE) are similarly run with default arguments, though we note that the network object used to define the distributed quantum system is a NISQNetwork, which allows us to set a limit of communication ancillas per QPU.
This limit is set to 1 communication ancilla per QPU.

The QFT circuit is defined in two different ways, depending on the Python module used for implementing the circuit partitioner.
FGP, GCP, MLFM are all included in DisQCO, which is based on IBM's Qiskit, while Pytket workflows are included in Pytket-DQC, which is based on Quantinuum's Pytket.
For Qiskit-based partitioners, the QPU circuit is first defined as an abstract circuit using the $\QFT$ gate in Qiskit's circuit library, then decomposed into $\U(3)$ and $\CP(\theta)$ gates which gives a circuit equivalent to the traditional circuit (see Fig.~\ref{fig:traditional-qft}) but with the $H$ gates represented as an appropriately parametrized $U(3)$ gate.
This transpilation is a requirement to use the partitioners as implemented in DisQCO.
For Pytket-based partitioners, the QPU circuit is directly defined as the sequence of gates of the traditional QFT (see Fig.~\ref{fig:traditional-qft}), which is already the appropriate set of gates to run the Pytket partitioning workflows (which accept only circuits made of $\H$, $\CP$, and $\RZ$ gates).

\subsection{Comparison of e-bit Counts from Our Schemes vs. from Circuit Partitioners} \label{sec:empirical-results}

E-bit counts and runtimes when partitioning into few QPUs are shown in Fig.~\ref{fig:few-QPUs-e-bits} and Fig.~\ref{fig:few-QPUs-time}, respectively.

First, we analyze the two-QPU case.
We notice that the \mbox{GCP-S} partitioner results in e-bit counts far exceeding the results of other partitioners.
GCP-S is a telegate-based partitioner which does not consider gate packing, lending credence to the idea that gate packing is a requirement for optimal partitioning of the QFT circuit.
Second, we notice that the FGP-rOEE partitioner obtains e-bit counts with an overhead factor of between $2\times$ and $3\times$ compared to more recent partitioners.
FGP-rOEE is a teledata-based partitioner, suggesting that telegates are a better fit for partitioning the QFT circuit.
Third, we notice that our detached-gate-based distributed QFT gives a lower bound on e-bit counts compared to all other partitioners, except for the
MLFM-R partitioner which obtains equal e-bit counts for two QPUs, and with the exception of Pytket-PHE on a 96\dash{}qubit QFT circuit across two QPUs.
This last exception is due to the Pytket workflow approximating the QFT circuit, which we consider an outlier result, as will be discussed in Section~\ref{sec:outliers}.

We interpret these results as hinting at the optimality of our gate-packed QFT for arbitrary bipartitions.
This is advantageous, since due to the fact that our gate-packing-based circuit partition is a pre-defined partitioning scheme, it can be implemented instantly (only taking time to construct the circuit object according to the scheme), while a circuit partition that results from a circuit partitioner takes time to run (in addition to the time it takes to construct the circuit object).
Indeed, as shown in Fig.~\ref{fig:few-QPUs-time}, the MLFM-R and Pytket partitioners may take a \mbox{non-negligible} amount of time to partition large QFT circuits, and are especially slow when partitioning into more QPUs.
As \citeauthor{burt_multilevel_2026}~\cite{burt_multilevel_2026} theorize, the reason MLFM-R performs as well as the gate-packed QFT is that it seems the initial attempt MLFM-R makes at partitioning, for any input circuit, just so happens to be equivalent to the gate-packed QFT.

We present similar results for a 4-QPU distributed system.
However, in this case, both the Pytket-PE and Pytket-PHE partitioners achieve much lower e-bit counts, lower than the gate-packed QFT, even for small circuit sizes.
We believe this is due to the ability of Pytket to consider detached gates at the expense of additional communication ancillas.
Note that all other circuit partitioning schemes and circuit partitioners are restricted to 1 communication ancilla per QPU (or 2 in the case of the linear-topology-based distributed QFT).
We discuss this further in Section~\ref{sec:outliers}.

\subsection{Pytket's E-bit Counts} \label{sec:outliers}

In Section~\ref{sec:empirical-results}, we saw that for two QPUs and for a large number of qubits, the Pytket-PHE partitioner is able to find an e-bit count lower than any other partitioner, including our detached-gate-based distributed QFT.
The reasons are twofold: consideration of detached gates and approximating the QFT.

In all workflows, the Pytket partitioners assume that the QPUs have an unbounded number of communication qubits, which allows them to use detached gates.
Detached gates are a powerful method for reducing e-bit requirements at the cost of an increase in communication qubits, as seen in Section~\ref{sec:content-proposed-dqft-detached}.

Separately, the Pytket partitioners approximate the QFT as part of their workflow.
As seen in Fig.~\ref{fig:traditional-qft}, the QFT circuit is implemented using $\CP(\theta)$ gates with exponentially small angles $\theta=2\pi/2^{j}$.
Because of how the Pytket partitioners are implemented in Python, angles very close to multiples of $2\pi$ are rounded to the nearest multiple of $2\pi$.
These approximated $\CP$ gates are then equivalent to the identity and are removed by the transpiler pass included in Pytket.
This creates an approximate QFT circuit, which has reduced connectivity compared to the fully connected exact QFT circuit. 
Therefore, fewer e-bits may be required to distribute it.
It is important to note that while the approximate QFT is a useful circuit to use in practice~\cite{coppersmith_approximate_1994,nam_approximate_2020}, it is unreasonable to compare a partitioned approximate QFT to a partitioned exact QFT; for a reasonable comparison, all partitioners should be benchmarked on the approximate QFT circuit, and our analytical gate-packed method should be extended to the approximate QFT.
Such an analysis is outside the scope of this paper.

\section{Implementation} \label{sec:content-implementation}

In this section, we discuss running our gate-packed QFT circuit on simulators and on hardware.
Firstly, we discuss the methodology for mapping the circuit and distributed quantum system to a specific monolithic QPU topology, so that it can be run on existing quantum hardware.
Secondly, we discuss the methodology for calculating the fidelity of the QFT calculation.
Lastly, we analyze the results and compare them to previous results in the literature.

\subsection{Mapping to Distributed Quantum System Topology}

In our approach, we have been partitioning the QFT as an abstract circuit, with each QPU assumed to be fully connected and each qubit assumed to be equivalent; this means the partitioned QFTs we have been defining are also abstract circuits.
As in the non-distributed setting, the abstract circuit needs to be transpiled to the quantum system's topology.
While the monolithic setting only needs to concern itself with the topology constraints of a single QPU, there are additional constraints when transpiling a partitioned circuit to a distributed quantum system.
These constraints, involving both the communication qubits (used for entangling QPUs) and the data qubits (used for holding and processing quantum information), are as follows:
\begin{itemize}
    \item The communication qubits defined in the abstract partitioning must be mapped to physical communication qubits as defined by the network's topology.
    \item The only two-qubit quantum operation applied between communication qubits is entanglement generation (e-bit generation).
    \item All data qubits assigned to QPU $k$ by the partitioning scheme must be mapped to the physical qubits of QPU $k$.
\end{itemize}
To our knowledge, there is no way to use standard publicly available transpilation passes (e.g., SabreLayout~\cite{zou_lightsabre_2024, noauthor_sabrelayout_nodate}) while specifying the above constraints on the layout mapping (for the communication qubits only)~\cite{wagner_re_2025}.
This implies that we have two choices: we either (1) let the entire transpilation happen automatically, likely breaking the requirement that logical communication qubits get mapped to physical communication qubits, or (2) manually set a mapping ourselves that respects the topology constraints of our distributed quantum system, likely giving a poorly optimized circuit.

We choose to manually set the mapping ourselves, but to remedy the problem of obtaining a possibly poorly optimized circuit, we generate multiple different mappings manually and pick the one giving the most optimized circuit after transpilation.
This manual mapping approach is detailed in~\cite{silver_q-dice_2026} as part of a pipeline for distributed quantum circuit emulation; for completeness, we describe it informally here.
The initial step of this mapping algorithm is to generate a set of various mappings that respect the topology constraints of our distributed quantum system.
This step involves first mapping the communication qubits as required, then randomly mapping the remaining qubits according to the QPU map, but with the additional heuristic that data qubits are mapped near the communication qubits.
Once this set is established, we evaluate each mapping according to the resulting depth of the transpiled circuit.
Circuit depth is a simple yet useful proxy for noise in quantum circuits, as deeper circuits involve more noisy gates and a longer time for decoherence to manifest.
Finally, we pick the mapping with the smallest circuit depth for the circuit we want to run.

To build the distributed quantum system in a way that we can run simulations and map to hardware, we embed its topology into that of a monolithic IBM QPU.
We call this a \emph{virtually partitioned} QPU, as our layout and qubit-mapping approach allows a single QPU to emulate a distributed quantum system.
This is equivalent to the \emph{QPU slicing} approach of~\cite{silver_q-dice_2026}, except that this work does not include a noise injection phase.

Specifically, we used the IBM Brisbane backend provided in Qiskit, virtually partitioned into two QPUs.
The Brisbane backend is an Eagle~r3 QPU, with early access to the improved dynamic circuit~\cite{ibm_new_2025}.
Dynamic circuits are circuits which include mid-circuit measurement and classical gate control, which is the mechanism used to define the circuit partitioning protocols of teledata and telegate.
Therefore, the improvements to dynamic circuits in the Eagle~r3 architecture---which involve faster execution (therefore less decoherence noise) and "improved mid-circuit measurements"%
\footnote{There are no details available for what "improved mid-circuit measurements" means on IBM Quantum's product update page~\cite{ibm_new_2025}.}%
---imply a potentially higher fidelity for running partitioned circuits.

\begin{figure}[!t]
    \centering
    \resizebox{0.6\linewidth}{!}{%
\begin{tikzpicture}
    \node[anchor=south west,inner sep=0] (image) at (0,0) {\includegraphics[width=\textwidth]{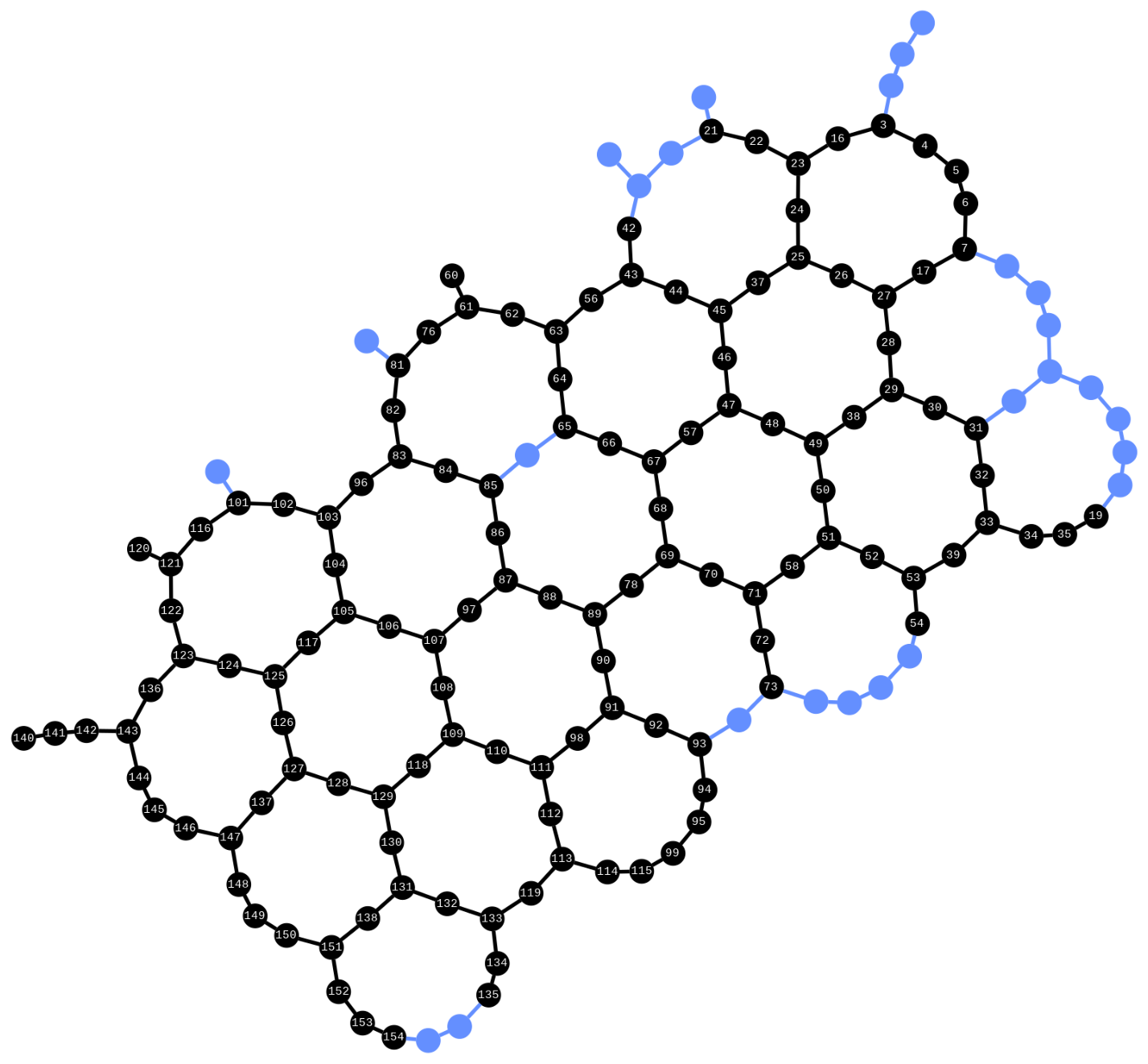}};
    \begin{scope}[x={(image.south east)},y={(image.north west)}]
        \draw[loosely dashed,red, ultra thick,decorate,decoration={snake, amplitude=.8mm, segment length=8mm, post length=0mm}] (0.29,0.84) -- (0.78,0.17);
    \end{scope}
\end{tikzpicture}
}
    \caption[Virtual Partition of a QPU Layout]{A layout of a partitioned QFT across a virtually split backend. The red line indicates a virtual partition. Only the middle qubits are connected as communication qubits.}
    \label{fig:split-backend}
\end{figure}

\subsection{Average Fidelity Calculation Methodology} \label{sec:fidelity-methodology}

We calculate the accuracy of a physical quantum circuit by averaging the fidelity of the noisy and ideal output states for some set of input states.

We independently conceptualized a methodology similar to \citeauthor{baumer_quantum_2024-1}'s methodology~\cite{baumer_quantum_2024-1} (see Section~\ref{sec:baumer-methodology}) to compute the average fidelity of a partitioned QFT circuit over the computational basis.

For completeness, we describe the circuit we use, based on the methodology described in Section~\ref{sec:baumer-methodology}, with terminology specific to the partitioned QFT.
To calculate the output fidelity 
\begin{equation}
    F\left(\QFT_n \ket{x}, \QFT^{(\text{distrib.})}_n \ket{x} \right)
\end{equation}
for $x \in \set{0, \dots, 2^n-1}$, where $\QFT_n$ is the ideal QFT circuit and $\QFT^{\text{(distrib.)}}_n$ is the distributed QFT circuit, we perform the following steps:
\begin{enumerate}
    \item Preparing the input state $\ket{x} = \ket{x_n} \dots \ket{x_0}$ with the circuit $\left(\X^{x_n} \tensor \dots \tensor \X^{x_0} \right)$ applied to the initial state $\ket{0^n}$.
    \item Running the noisy distributed circuit $\QFT^{\text{(distrib.)}}_{n}$. 
        We now have $\QFT^{\text{(distrib.)}}_{n} \ket{x}$.
    \item Running the inverse of the ideal version of the circuit $\QFT_n^\dagger$ as it would be applied to the input state, and counting the output $0$ in the computational basis.
        This is applying 
        \begin{align*}
            \bra{x} \QFT_n^\dagger &= \bra{0^n} \H^{\tensor n}  \left( \P(-2\pi 0.x_0) \right. \\
            &\phantom{= \bra{0^n} \H^{\tensor n}} \tensor \P(-2\pi 0.x_1 x_0) \\
            &\phantom{= \bra{0^n} \H^{\tensor n}} \tensor \dots \\
            &\phantom{= \bra{0^n} \H^{\tensor n}} \left. \tensor \P(-2\pi 0.x_n\dots x_0) \right) \;.
        \end{align*} 
    \item Repeating these steps a sufficient number of times to get a good estimate for the probability of measuring the~$\ket{0^n}$~state, which corresponds to our fidelity.
\end{enumerate}

While~\textcite{baumer_quantum_2024-1} do a similar procedure by sampling \mbox{$x \in \set{0 \dots 2^n-1}$} uniformly, we choose to sample all $2^n$ input states in our basis the same number of times.
This is more accurate but less computationally efficient; however, this inefficiency does not affect us as we will only need to compute results for small $n$ (as large $n$ will be less noisy).

\subsection{Fidelity Results}

\begin{figure}[!t]
    \centering
    \includegraphics[width=0.5\linewidth]{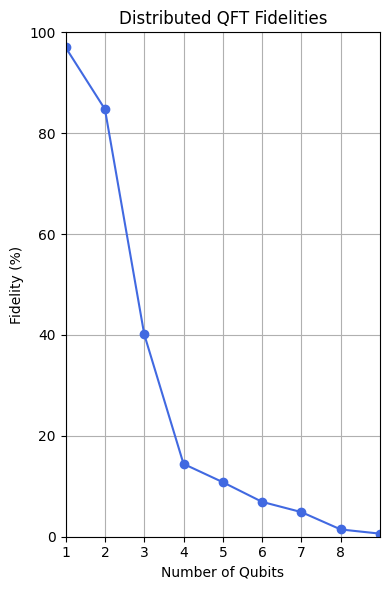}
    \caption[Simulated Average Fidelities]{Simulated average fidelity of the gate-packed QFT partitioned across two QPUs. The average fidelity is computed according to the methodology of Section~\ref{sec:fidelity-methodology}. The QFT sizes increase until the average fidelity is below 1\%, here at 9 qubits (excluding the communication ancillas). The simulation was done using the AerSimulator for the Brisbane backend, which mimics the noise model of the IBM Brisbane QPU. }
    \label{fig:simulated-fidelities}
\end{figure}

In Fig.~\ref{fig:simulated-fidelities}, we see the simulated average fidelities, while in Fig.~\ref{fig:hardware-fidelities}, we see the same as implemented on a quantum computer.
We can compare this to~\cite[Fig.~2~(a)]{baumer_quantum_2024-1}, which is a similar plot but with a hardware implementation on IBM Kyiv of four variants of the QFT circuit: the traditional QFT circuit with and without dynamical decoupling (DD) and the semiclassical QFT with and without DD.
We note that our simulation performs similarly to their traditional QFT without dynamical decoupling, and so worse than the three other variants.
It is not surprising that both our gate-packed QFT circuit partitioning and their traditional QFT are worse than their semiclassical QFT, as both the traditional QFT and the gate-packed distributed QFT are coherent and unavoidably require many two-qubit gates. 
Furthermore, it is unsurprising that our gate-packed partitioned QFT and their no-DD traditional QFT perform worse than their QFT with DD error suppression; dynamical decoupling and other error suppression strategies have the potential to improve the average fidelity of our gate-packed partitioned QFT, but this was outside the scope of our research.

We emphasize that our gate-packed QFT circuit partitioning has an advantage over their traditional QFT, which is that it can scale by adding more QPUs, while their traditional QFT is bounded by the size of its single QPU.
Furthermore, our gate-packed QFT circuit partitioning can be used in applications where the result needs to stay coherent, while their semiclassical QFT includes measuring all qubits at the end.

\begin{figure}[!t]
    \centering
    \includegraphics[width=0.6\linewidth]{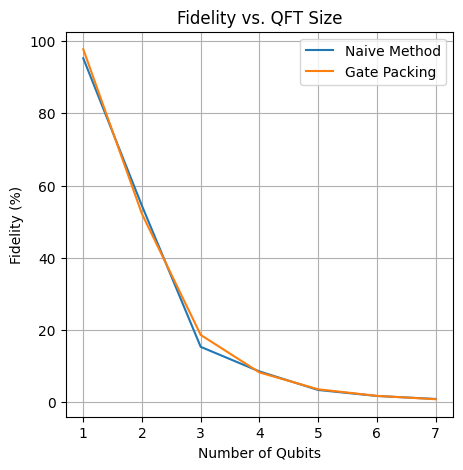}
    \caption[Average Fidelities on Virtually Partitioned IBM QPU]{Average fidelities of the naively partitioned QFT and our gate-packed partitioned QFT across two virtual QPUs on the IBM Brisbane QPU, an Eagle~r3 QPU with early access to the new dynamic circuits. The average fidelity is computed according to the methodology of Section~\ref{sec:fidelity-methodology}. The QFT sizes increase until the average fidelity is below 1\%, here at 7 qubits (excluding the communication ancillas).}
    \label{fig:hardware-fidelities}
\end{figure}

We also proceed with an analysis of the fidelity variance in Fig.~\ref{fig:fidelity-nonaveraged}.
These data are from the same experiment as for average fidelities (Fig.~\ref{fig:hardware-fidelities}), but isolated to the $4$\dash{}qubit gate-packed QFT.
We see that for each input $\ket{x}$, $x \in \set{0, \dots, 2^4-1}$, the standard deviation is small, meaning that our choice of 200~samples per input $x$ is sufficient for the purposes of estimating the fidelity for each input.
Furthermore, we see that across inputs, the fidelity remains rather stable, hovering around $8\%$.
Separately, in Fig.~\ref{fig:fidelity-middle}, we look at the fidelity across all QFT sizes $n$, but isolating to the middle input $x=2^{n-1}$.
We see that the standard deviation remains similarly sized across all QFT sizes, showing that the choice of 200~samples per input $x$ to estimate the output fidelities does not need to be adjusted as the circuit size increases. 

\begin{figure}[!t]
    \centering
    \includegraphics[width=0.7\linewidth]{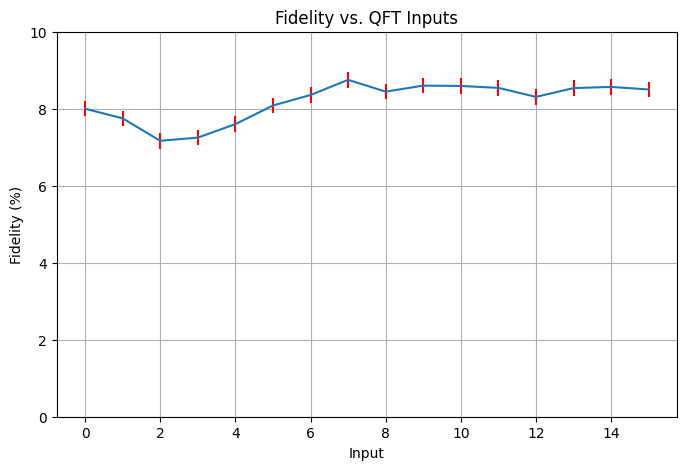}
    \caption[Variance of Fidelity Across Inputs]{Output fidelities of the $4$\dash{}qubit gate-packed QFT on the IBM Brisbane QPU. The x\dash{}axis corresponds to the QFT input $\ket{x}$, where $x \in \set{0, \dots, 2^4-1}$. The y\dash{}axis corresponds to the fidelity computed with the methodology of Section~\ref{sec:fidelity-methodology}, and each value is an average over 200 measurement results. Note that the y\dash{}axis is scaled to the range of $0-10\%$. The error bars in red correspond to one standard deviation. }
    \label{fig:fidelity-nonaveraged}
\end{figure}

\begin{figure}[!t]
    \centering
    \includegraphics[width=0.6\linewidth]{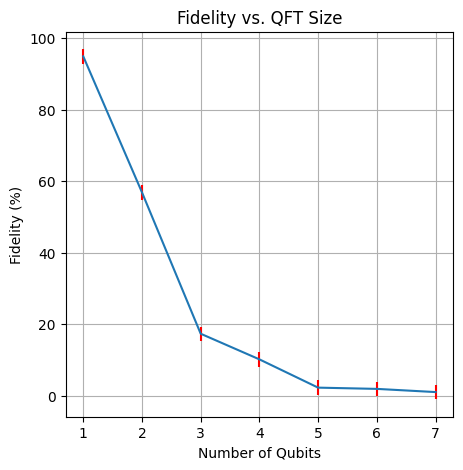}
    \caption[Variance of Average Fidelities]{Output fidelities of the gate-packed partitioned QFT for input $\ket{x} = \ket{2^{n-1}}$ across two virtual QPUs on the IBM Brisbane QPU (an Eagle~r3 QPU). The QFT sizes increase until the output fidelity is below 1\%. The error bars in red correspond to one standard deviation.}
    \label{fig:fidelity-middle}
\end{figure}

We also ran the naively partitioned QFT circuit on a hardware QPU alongside our gate-packed partitioned QFT circuit, with results reported in Fig.~\ref{fig:hardware-fidelities}.
We note that the results are very similar, though the gate-packed QFT exhibits a slightly higher fidelity for the three-qubit QFT.
We suspect that this negative result is due to the large noise presence on the QPU even in the absence of circuit partitioning~\cite{baumer_quantum_2024-1}, especially with the large number of gates that the QFT needs to run.
This large amount of noise---derived from coherent two-qubit gates---may be dominating the noise from mid-circuit measurement and control (from the starting and ending processes), giving similar fidelity results no matter the number of e-bits used to implement the non-local operations.
Further research is needed.

We see in Fig.~\ref{fig:hardware-times} that despite the similarity in average fidelities, the e-bit savings of the gate-packing-based method translate to time savings in running the QFT. 
This is due to e-bit counts being directly proportional to invocations of quantum teleportation protocols, which are far slower than coherent quantum operations and so dominate circuit runtime.
The circuit speedup is especially evident for an increasing number of qubits; as shown in Section~\ref{sec:comparison-yim}, the e-bit count for the naive scheme scales quadratically in the number of qubits, compared to linearly for our gate-packing-based scheme.

\begin{figure}[!t]
    \centering
    \includegraphics[width=0.7\linewidth]{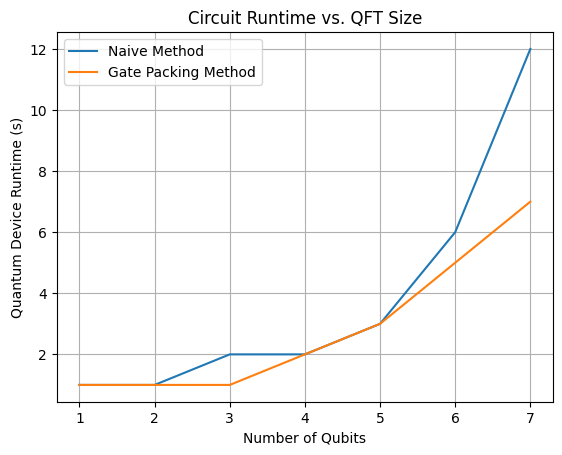}
    \caption[Time to Compute Average Fidelities on IBM QPU]{Time to compute the average fidelity of the naively partitioned QFT and of our gate-packed partitioned QFT across a virtually 2-partitioned IBM Brisbane QPU, for increasing QFT size. The average fidelity is computed according to the methodology of Section~\ref{sec:fidelity-methodology}. Note that the data are only made available by IBM to 1 second of precision.}
    \label{fig:hardware-times}
\end{figure}

\section{Conclusions} \label{sec:conclusions}

In this work, we minimized entanglement resources for the distributed quantum Fourier transform (QFT) using teleportation-based gate packing. Specifically, we focused on reducing total e-bit consumption when running the QFT circuit on a distributed quantum system using gate-packed implementations of the QFT circuit, subject to (i) fixed QPU decomposition, (ii) teleportation-based non-local gates, and (iii) a communication budget per QPU, which we vary. Under these constraints, we explored two operational regimes: a strictly bounded ancilla regime (one communication ancilla per QPU) and an unbounded ancilla regime leveraging detached gates.

Our main contribution is a characterization of the minimized entanglement cost under these constraints. In the bounded ancilla setting, we analytically proved a tight lower bound on the required number of e-bits based on the QFT circuit's qubit interaction structure, and we presented a circuit construction scheme that achieves this bound. Central to this result is the observation that, for the standard QFT decomposition, entanglement consumption is minimized when qubits are assigned to processing units in a manner that respects the intrinsic ordering of two-qubit interactions, ensuring that non-local gates are introduced only when required by the circuit structure. Furthermore, for the unbounded ancilla setting, we introduced an optimal hub-based partitioning scheme using detached gates and proved via hypergraph analysis that it achieves a strictly lower e-bit bound for multipartite distributions. Together, these results establish that no gate-packing-based distributed implementation of the exact QFT can asymptotically improve upon our constructions without relaxing our structural assumptions.

While our gate-packed QFT circuit achieves the coherent communication complexity lower bound in the bipartite case, detached gates allow for further asymptotic e-bit reductions in the multipartite regime when additional communication ancillas are available. Our results provide bounded, provably optimal solutions across two complementary compilation frameworks that align with current teleportation-based distributed quantum architectures.

Beyond the specific case of the QFT, our analysis highlights a more general structural insight: circuits whose interaction graphs exhibit sequential or triangular connectivity admit entanglement-optimal distributed structures via similar partitioning strategies. This suggests that similar optimization techniques may apply to other quantum algorithms with comparable interaction patterns, such as quantum walks and arithmetic, and motivates the study of interaction-graph-aware compilation as a key principle in optimizing distributed quantum algorithms.

In addition to our theoretical analysis, we conducted an empirical evaluation of several general-purpose circuit partitioning compilers applied to distributed QFT instances. These experiments demonstrated strong evidence for the effectiveness of our analytical methods, with the best partitioners yielding results that match our single-ancilla construction in the bipartite case, while providing e-bit counts much higher than our detached-gate-based scheme in the multipartite case. This reinforces the necessity of structure-aware optimization in future works. Our results thus position our methods as practical benchmarks against which future distributed compilation heuristics can be evaluated.

We further validated our approach by implementing distributed versions of the QFT on both simulated and physical quantum platforms. Leveraging a custom software pipeline for executing distributed circuits, we deployed our gate-packed QFT on IBM’s quantum devices and evaluated its performance by analyzing the average fidelity across the computational basis. These implementations demonstrate the feasibility of the proposed partitioning strategy under realistic constraints, the primary one being ancilla availability. While the primary focus of this work is entanglement optimization rather than noise resilience, experimental results confirm, within the scope of these experiments, that the entanglement footprint achieved by our construction does not introduce observable degradation in hardware implementation. Rather, we attribute the sharp fidelity drop to hardware noise native to highly connected circuits. This practical demonstration not only validates the correctness of our circuit but also provides a clear methodology for evaluating deployable distributed quantum algorithms on today's hardware.

Finally, this work opens several directions for future research. For one, the scope of QFT circuits under consideration can be expanded, including, for example, the fast parallel circuit of \citeauthor{cleve_fast_2000}~\cite{cleve_fast_2000}, which has a very different structure from all QFT circuits studied here. Alternatively, the core of this paper can be repeated for the approximate QFT~\cite{coppersmith_approximate_1994}, which has a different structure from the exact QFT circuit, and which might be more amenable to current QPU device topologies~\cite{cleve_fast_2000, baumer_approximate_2025}. Furthermore, the implementations on quantum hardware (Section~\ref{sec:content-implementation}) can be expanded to all partitioned QFT circuits under consideration (described in Sections~\ref{sec:other-distributed-qft}, \ref{sec:proposed-dqft}, and \ref{sec:proposed-dqft-2}) for completeness, and dynamical decoupling error mitigation can be used as in~\cite{baumer_quantum_2024-1} to improve circuit fidelity.

%% file: appendix.tex
\appendices

\section{Proof of Optimal QPU Mapping for Gate-Packed Distributed QFT}

Let $\lastactive: k \mapsto \max\phi^{-1}(k)$ be the mapping of QPU $k$ to the qubit assigned to it that will be active last.
Therefore, for time slice $j$, the number of active QPUs is $a_j = \abs{\set{k \in \set{1 \dots m} \st \lastactive(k) \leq j}}$.
Furthermore, without loss of generality, assume QPUs indices $1,\dots,m$ correspond to being ordered by $\lastactive$, i.e., $\lastactive(1) \leq \lastactive(2) \leq \dots \leq \lastactive(m)$.

We start by proving the following lemma:
\begin{lemma}\label{lemma:optimal-e-bits}
The number of e-bits required for a gate-packed distributed QFT is
\begin{equation} \label{eq:lemma-optimal-e-bits}
    E = \sum_{k=1}^{m-1} \lastactive(k) \;.
\end{equation}
\end{lemma}

\begin{proof}
    Now, the sequence $a_1 \dots a_n$ of number of active QPUs per time slice is of the form
    \begin{align*}
        a_{1}, \dots, a_{\lastactive(1)} &= m \\
        a_{\lastactive(1)+1}, \dots, a_{\lastactive(2)} &= m-1 \\
        \cdots \\
        a_{\lastactive(m-2)+1}, \dots, a_{\lastactive(m-1)} &= 2 \\
        a_{\lastactive(m-1)+1}, \dots, a_{\lastactive(m)} &= 1 \;.
    \end{align*}
    This is true since we know that all QPUs start as active ($a_1 = m$) and that QPUs deactivate sequentially after their last active qubit becomes inactive ($a_{\lastactive(k)+1} = a_{\lastactive(k)}-1$).
    Using this fact to group $a_j$ by their values in Equation~\eqref{eq:e-bit formula active} and telescoping the sum, we get 
    \begin{equation*}
         E = \sum_{k=1}^{m-1} \lastactive(k).
    \end{equation*}
\end{proof}

We now move to the main theorem.

\setcounter{theorem}{0}
\begin{theorem}[restated]
    Let  $QPU_k$ be the $k$th smallest QPU. 
    The optimal assignment is for $QPU_1$ to be assigned the first (most significant) $n_1$ qubits, $QPU_2$ the next $n_2$ most significant qubits, etc. until $QPU_m$ is assigned the $n_m$ least significant qubits.
    This can be equivalently stated as $QPU_k$ being assigned qubits $$\phi^{-1}(k) = \set{n-\sum_{j=m}^{k+1}n_{j}, \dots, n - \sum_{j=m}^{k}n_{j}+1} \;. $$
\end{theorem}

\begin{proof}
    As in Lemma~\ref{lemma:optimal-e-bits}, assume QPUs are ordered by $\lastactive$, i.e., the map $k \mapsto \lastactive(k)$ is monotonically increasing.
    
    The structure of the proof is as follows.
    First, we present a general procedure for arbitrarily assigning qubits to QPUs.
    Second, we rewrite the formula of Lemma~\ref{lemma:optimal-e-bits} to match the given procedure.
    Third, we show that minimizing this formula gives an optimal assignment.
    Finally, we show that this optimal assignment has $QPU_k$ being the $k$th smallest QPU.

    Consider the following general procedure for assigning qubits, which will proceed iteratively by assigning qubits in $m$ steps $m, m-1, \dots, 1$, starting with the qubits assigned to QPU $m$ and finishing with qubits assigned to QPU $1$.
    
    Let $\Rem_k$ be the set of qubits that remain unassigned at step $k$, after having already assigned qubits to QPUs $m, m-1, \dots, k+1$.

    Since both qubits and QPUs are indexed by order of becoming inactive, we know the last active qubit not in QPUs $k+1, \dots, m$ will be $\max\Rem_k$ and conversely the last active qubit of QPU $k$ will be
    \begin{equation} \label{eq:lastactive is max}
        \lastactive(k) = \max \Rem_{k} \;.
    \end{equation}
    Furthermore, any other qubits in $\phi^{-1}(k)$ will not be the last active qubit of this QPU. 
    Let these $n_{k}-1$ other qubits be chosen arbitrarily from the currently unassigned qubits $\Rem_{k} \setminus \set{\lastactive(k)}$.
    Let the remaining unassigned qubits after choosing $\phi^{-1}(k)$ be defined as $\Rem_{k-1} :=  \Rem_{k} \setminus \phi^{-1}(k) $.

    Since we know the base case $\Rem_{m} = \set{1,\dots,n}$, the above is sufficient for an induction.

    Now, by Lemma~\ref{lemma:optimal-e-bits}, we need to minimize $E = \sum_{k=1}^{m-1} \lastactive(k)$.
    By Equation~\ref{eq:lastactive is max} of the above induction, we therefore need to minimize
    \begin{equation*}
        E = \sum_{k=1}^{m-1} \max \Rem_k \;.
    \end{equation*}
    
    Suppose we were to minimize $\max\Rem_k$ independently for each $k \in \set{1, \dots,m-1}$ (ignoring $k=m$ since we know $\max\Rem_m = n$).
    For any given $k$, since $\Rem_k \subseteq \set{1, \dots, n}$ we have the lower bound $\max \Rem_k \geq \abs{\Rem_k}$ which is attained by 
    \begin{equation*}
        \Rem_k := \set{1, \dots, \abs{\Rem_k}} \;.
    \end{equation*}
    Since $\Rem_k \subseteq \Rem_{k+1}$, these locally optimal choices are all compatible and therefore produce a globally optimal assignment $\phi^{-1}(k) = \Rem_k \setminus \Rem_{k-1}$ (where $\Rem_0 := \emptyset$):
    \begin{align*}
        \phi^{-1}(m) &= \set{n,\dots,n-n_{m}+1} \\
        \phi^{-1}(m-1) &= \set{n-n_{m}, \dots, n-n_{m}-n_{m-1}+1} \\
        &\cdots \\
        \phi^{-1}(k) &= \set{n-\sum_{j=k+1}^{m}n_{j}, \dots, n - \sum_{j=k}^{m} n_{j}+1} \\
        &\cdots \\
        \phi^{-1}(1) &= \set{ n-\sum_{j=2}^{m}n_{j}, \dots, 1} = \set{ n_{1}, \dots, 1} .
    \end{align*}
    
    With this assignment, the e-bit count is
    \begin{align*}
        E &= \sum_{k=1}^{m-1} \max \Rem_k \\
            &= \sum_{k=1}^{m-1} \left( n - \sum_{j=k+1}^{m} n_j \right) \\
            &= (m-1)n - \sum_{k=1}^{m-1} \sum_{j=k+1}^{m} n_{j} \\
            &= (m-1)n - \sum_{k=1}^{m-1} k n_k \;.
    \end{align*} 
    
    Now, all that remains is fixing the assignment of $k \mapsto n_{k}$, i.e., the ordering in which QPUs are made inactive first according to the above scheme. 
    To minimize the e-bit count, we must choose $n_{m}$ to be the size of the largest QPU, then $n_{{m-1}}$ the second largest, etc. Therefore, in addition to $QPU_k$ being the $k$th QPU to become inactive, it is also the $k$th smallest QPU.
\end{proof}

%% file: acknowledgements.tex
\section*{Acknowledgment}

We are grateful to Sean Wagner of IBM for his help with Qiskit transpilation and to Aharon Brodutch of IonQ for pointing us to~\citeauthor{rodney_van_meter_communications_2004}'s paper. 
Thanks to the IBM Quantum Support Team for their help getting us early access to the new dynamic circuit hardware.

This research was enabled in part by quantum device access provided by PINQ² (Quebec Digital and Quantum Innovation Platform) and by support provided by Compute Ontario and the Digital Research Alliance of Canada (\href{https://alliancecan.ca/}{alliancecan.ca}).

This research has received funding from the research project entitled “Quantum Software Consortium: Exploring Distributed Quantum Solutions for Canada” (QSC). QSC is financed by the National Sciences and Engineering Research Council of Canada (NSERC) Alliance Consortia Quantum program under grant number ALLRP587590-23.

This research was conducted in part with support for the first author by the Erasmus+: Erasmus Mundus programme of the European Union.

%% file: bibliography.tex
\printbibliography[heading=bibintoc]